\documentclass[journal]{IEEEtran}
\IEEEoverridecommandlockouts
\usepackage{cite}

\usepackage{amsfonts,amsmath,amssymb,amsthm,mathrsfs,bm,bbm}
\usepackage{graphicx}
\usepackage{overpic}
\usepackage{textcomp}
\usepackage{xcolor}
\usepackage{flushend}
\usepackage{enumerate}
\usepackage{algpseudocode}

\ifCLASSOPTIONcompsoc 
	\usepackage[caption=false,font=normalsize,labelfon t=sf,textfont=sf]{subfig} 
\else 
	\usepackage[caption=false,font=footnotesize]{subfig} \fi

\usepackage{hyperref}
\hypersetup{
  colorlinks   = true,    
  urlcolor     = black,    
  linkcolor    = black,    
  citecolor    = red      
}

\newcommand{\eqdef}{:=}
\newcommand{\reqdef}{=:}
\renewcommand{\vec}[1]{\mathbf{#1}}		
\newcommand{\rvec}[1]{\mathbbm{#1}} 		
\newcommand{\rmat}[1]{\mathbbm{#1}} 	
\newcommand{\E}{\mathsf{E}}		
\newcommand{\Var}{\mathsf{V}}			
\newcommand{\stdset}[1]{\mathbbmss{#1}}	
\newcommand{\set}[1]{\mathcal{#1}}		

\newcommand{\CN}{\mathcal{CN}}			
\newcommand{\herm}{\mathsf{H}}			

\newcommand{\refeq}[1]{(\ref{#1})}		%

\newtheorem{definition}{Definition}

\newtheorem{proposition}{Proposition}
\newtheorem{assumption}{Assumption}

\newtheorem{remark}{Remark}

\newcommand{\real}{\stdset{R}} 		

\newcommand{\Natural}{{\mathbb N}}

\begin{document}
\def\baselinestretch{.98}
\setlength{\belowdisplayskip}{0.9pt}
\setlength{\belowdisplayshortskip}{0.9pt}
\setlength{\abovedisplayskip}{0.9pt}
\setlength{\abovedisplayshortskip}{0.9pt}

\title{Two-timescale joint power control and beamforming design with applications to cell-free massive MIMO}

\author{Lorenzo Miretti~\IEEEmembership{Member,~IEEE}, Renato~L.~G. Cavalcante~\IEEEmembership{Member,~IEEE},\\ Sławomir Sta\'nczak~\IEEEmembership{Senior Member,~IEEE}
\thanks{This work was supported by the Federal Ministry of Education and Research of Germany in the programme of “Souverän. Digital. Vernetzt.” Joint projects 6G-ANNA and 6G-RIC, project identification numbers: 16KISK087, 16KISK020K, 16KISK030.

L.~Miretti was with the Fraunhofer Institute for Telecommunications Heinrich-Hertz-Institut HHI, 10587 Berlin, Germany, and also with the Technische Universität Berlin, 10587 Berlin, Germany. He is now with Ericsson Research, Germany (email: lorenzo.miretti@ericsson.com).

R.~L.~G. Cavalcante, and S.~Sta\'nczak are with the Fraunhofer Institute for Telecommunications Heinrich-Hertz-Institut HHI, Berlin 10587, Germany (email: \{cavalcante, stanczak\}@hhi.fraunhofer.de). S.~Sta\'nczak is also with Technische Universität Berlin, Berlin 10587, Germany. }
}

\maketitle


\IEEEpubid{\begin{minipage}{\textwidth}\ \\[12pt] \centering
© 2025 IEEE. Personal use of this material is permitted. Permission from IEEE must be obtained for all other uses, in any current or future media, including reprinting/republishing this material for advertising or promotional purposes, creating new collective works, for resale or redistribution to servers or lists, or reuse of any copyrighted  component of this work in other works.
\end{minipage}}

\begin{abstract} In this study we derive novel optimal algorithms for joint power control and beamforming design in modern large-scale MIMO systems, such as those based on the cell-free massive MIMO and XL-MIMO concepts. In particular, motivated by the need for scalable system architectures, we formulate and solve nontrivial two-timescale extensions of the classical uplink power minimization and max-min fair resource allocation problems. In our formulations, we let the beamformers be functions mapping partial instantaneous channel state information (CSI) to beamforming weights, and we jointly optimize these \textit{functions} and the power control coefficients based on long-term statistical CSI. This long-term approach mitigates the severe scalability issues of competing short-term iterative algorithms in the literature, where a central controller endowed with global instantaneous CSI must solve a complex optimization problem for every channel realization, hence imposing very demanding requirements in terms of computational complexity and signaling overhead. Moreover, our approach outperforms the available long-term approaches, which do not jointly optimize powers and beamformers. The obtained optimal long-term algorithms are then illustrated and compared against existing short-term and long-term algorithms via numerical simulations in a cell-free massive MIMO setup with different levels of cooperation. 
\end{abstract}


\section{Introduction}
\IEEEpubidadjcol
There is a growing consensus in both industry and academia that the physical layer of next-generation mobile communication networks will primarily rely on sophisticated multi-user multiple-input-multiple-output (MU-MIMO) systems for spectrally efficient multiple access in the lower (sub-6 GHz) and upper mid-band (7-24 GHz). The vast majority of these envisioned systems are rooted in the now established massive MIMO paradigm \cite{marzetta2016fundamentals,massivemimobook}, and, in particular, in its most ambitious embodiments such as those based on the concepts of cell-free massive MIMO \cite{ngo2017cell,demir2021} or extremely large-scale MIMO (XL-MIMO) \cite{debbah2023xl}. A core aspect of the above technologies is their potential to perform precise beamforming (also called beam focusing, in the context of near-field communications) and to efficiently allocate resources such as time, frequency, and power to guarantee unprecedented quality of service to all users in the network. Specifically, by deploying a large number of infrastructure antennas across the service area, these technologies are able to effectively multiplex many users in the same time-frequency resource using standard single-user channel codes, linear array processing (i.e., beamforming), and appropriate power control mechanisms \cite{marzetta2016fundamentals,massivemimobook, demir2021}. 

\begin{figure*}[!ht]
		\centering
		\subfloat[]{\includegraphics[width=0.67\columnwidth]{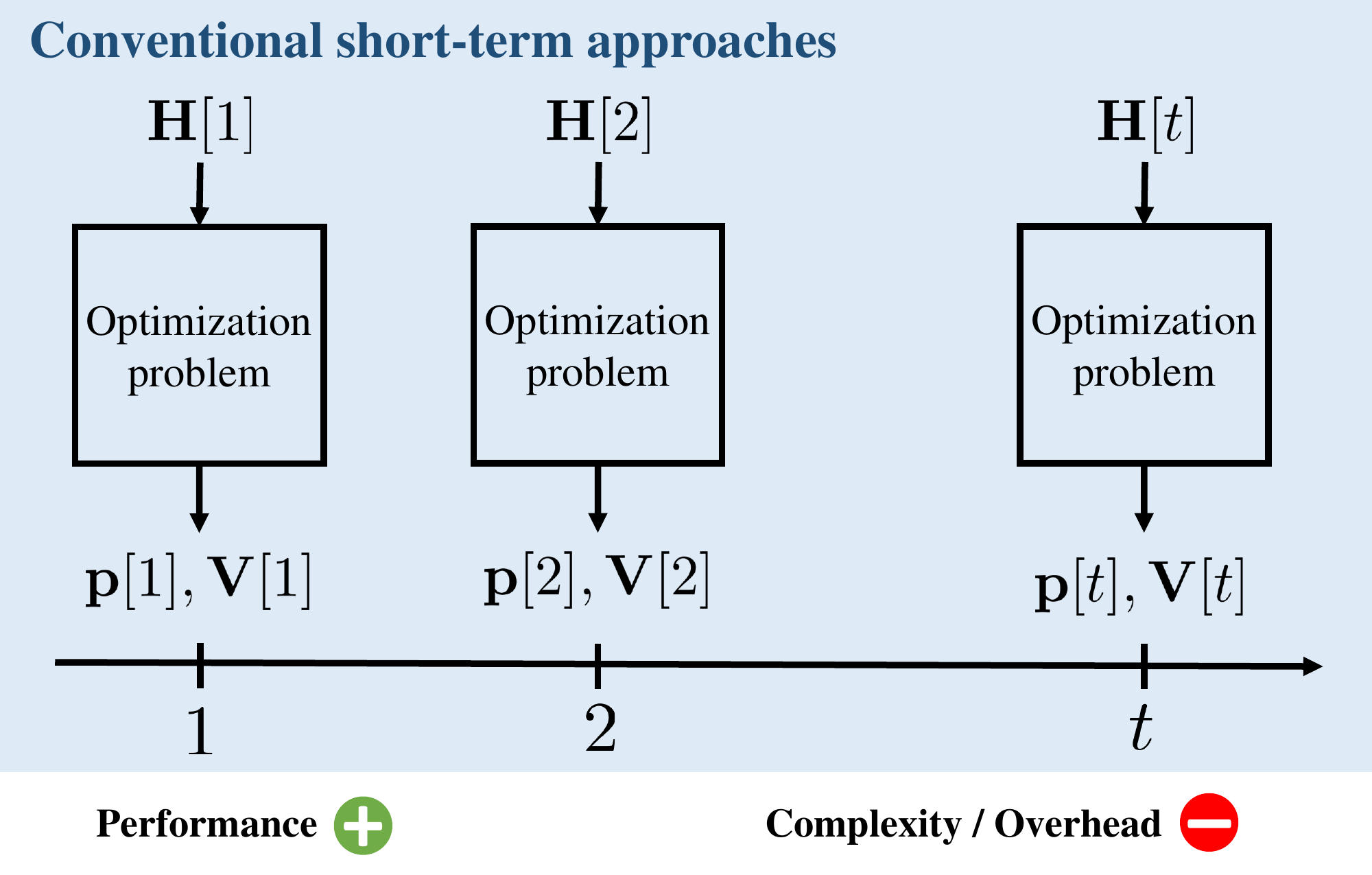}
			\label{fig:short}}
		\subfloat[]{\includegraphics[width=0.67\columnwidth]{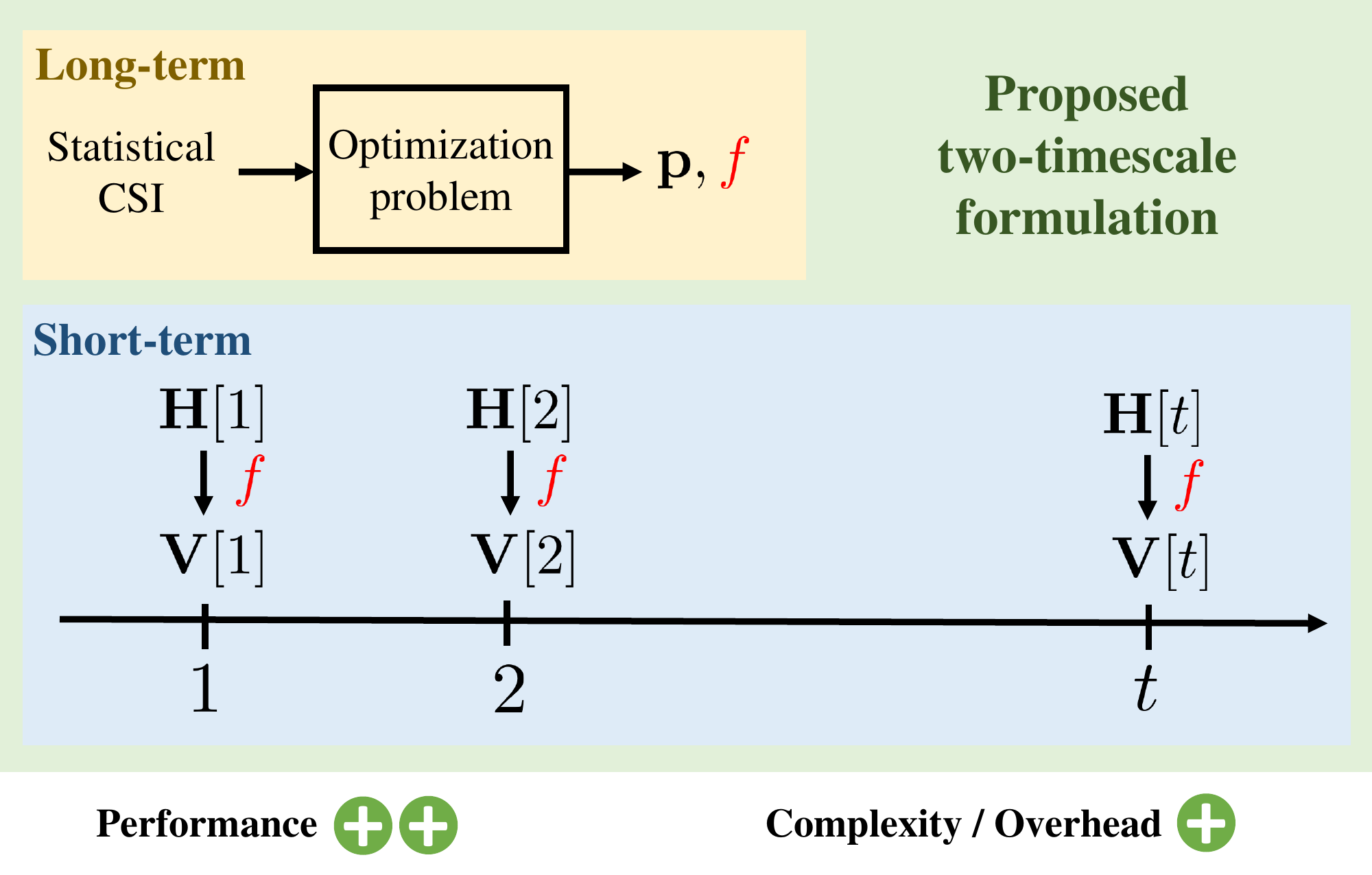}
			\label{fig:long}}
			\subfloat[]{\includegraphics[width=0.67\columnwidth]{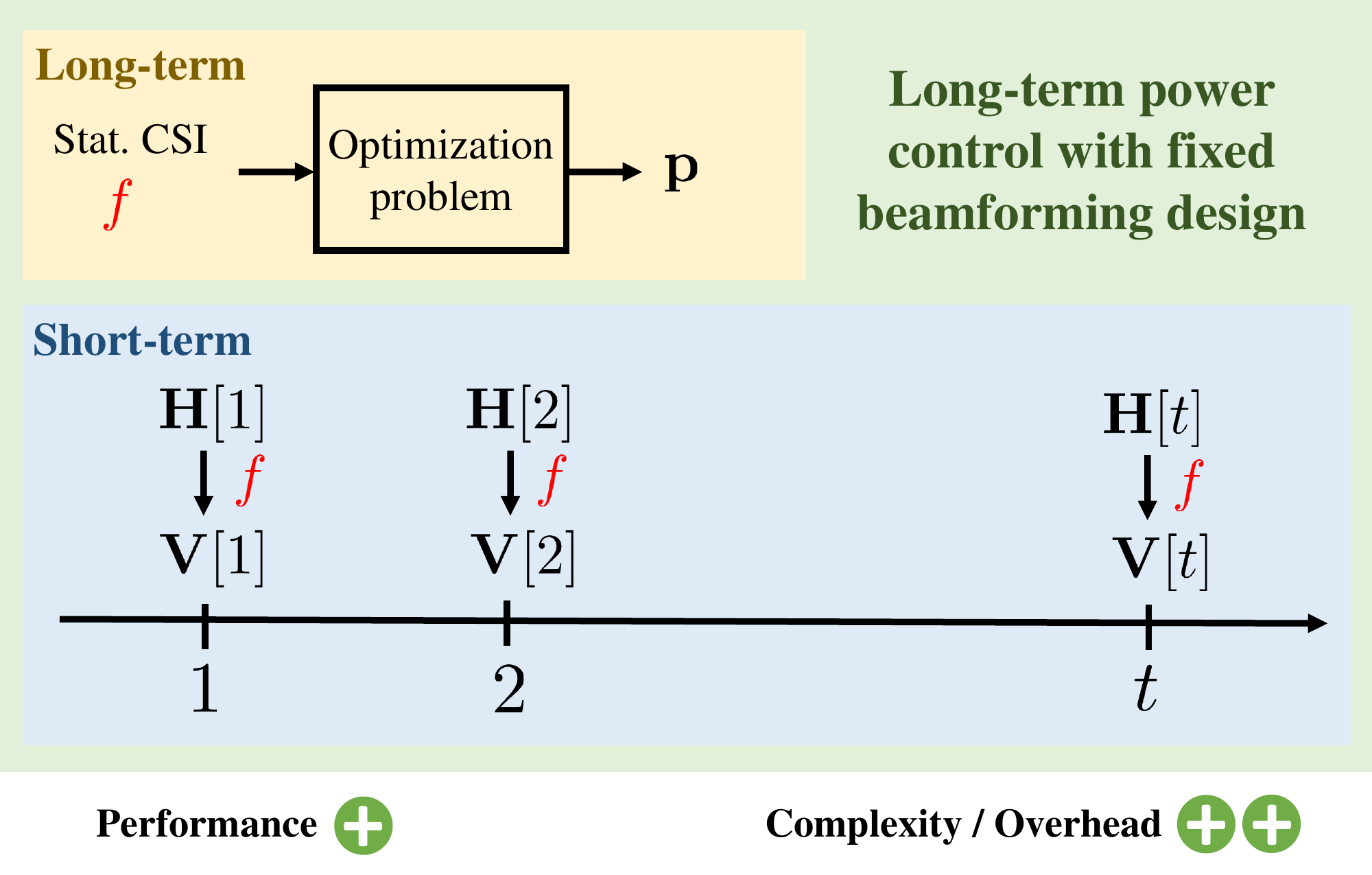}
			\label{fig:power}}
		\caption{Illustration of the differences between the proposed and existing approaches for the particular case of perfect CSI and max-min SINR criterion. In Fig.~(a), a complex optimization problem is solved for every $t$th channel realization $\vec{H}[t]$ to find jointly optimal instantaneous powers $\vec{p}[t]$ and beamformers $\vec{V}[t]=(\vec{v}_1[t],\ldots, \vec{v}_K[t])$. In contrast, in Fig.~(b), a single optimization problem is solved to find jointly optimal long-term powers $\vec{p}$ and beamforming function $f$ based on channel statistics. The obtained power vectors are kept fixed for all channel realizations, while the beamforming weights evolve over each channel realization as $\vec{V}[t]=f(\vec{H}[t])$. Fig.~(c) depicts a typical long-term power control approach in the literature, where the function $f$ is fixed \textit{a priori} (e.g., to maximum-ratio combining or zero-forcing), and it is not jointly optimized with $\vec{p}$. We note that, although introduced here for illustration purposes, in this work we do not use the functional notation $f$ and the channel realization index $t$, but we use a completely equivalent stochastic formulation.}
\label{fig:intro}
\end{figure*}

To capitalize on the full potential of the above technologies, the beamformers and the power control coefficients should be jointly optimized. However, the conventional joint optimization approaches borrowed from the classical MU-MIMO literature are known to be essentially inapplicable to modern large-scale MIMO systems, due to their severe scalability issues in terms of computational complexity and signaling overhead, both over the air interface and over the fronthaul or array interconnections. Therefore, most of the research on power control and beamforming design for cellular and cell-free massive MIMO systems is devoted to addressing the scalability issues of the aforementioned conventional approaches \cite{marzetta2016fundamentals,massivemimobook,ngo2017cell,nayebi2017precoding,bashar2019maxmin,bjornson2019making,du2021cellfree,demir2021,interdonato2021enhanced}. Nevertheless, despite the remarkable progress, all the available solutions essentially renounce jointly optimizing the beamformers and the power control coefficients, hence they provide suboptimal designs that can result in non-negligible performance loss, as verified later via numerical simulations.

In this study, we provide a novel optimal approach to joint power control and beamforming design that overcomes the main limitations of the conventional approaches from the classical MU-MIMO literature, and of the available approaches from the (cell-free) massive MIMO literature, as detailed next.
\IEEEpubidadjcol

\subsection{Conventional short-term approaches}
To jointly optimize the beamformers and the power control coefficients, two popular criteria are the sum power minimization subject to signal-to-interference-to-noise (SINR) constraints, and the maximization of the minimum weighted SINR subject to power constraints (sometimes also referred to as the SINR balancing problem) \cite{rashid1998joint, schubert2004solution, wiesel2006linear, yu2007transmitter,cai2011maxmin}. By focusing on the uplink of a system with $K$ single-antenna users and a base station equipped with $M$ antennas, these two optimization problems are typically  stated as 
\begin{equation}\label{eq:QoS_det}
		\underset{\substack{\vec{p},\vec{v}_1,\ldots,\vec{v}_K}}{\text{minimize}}
		\; \sum_{k\in\set{K}}p_k 
		\text{ s.t. } (\forall k\in\mathcal{K})~\mathsf{SINR}_k^{\mathsf{inst}}(\vec{H},\vec{v}_k,\vec{p}) \geq \gamma_k, 
\end{equation}
and
\begin{equation}\label{eq:maxmin_det}
		\underset{\substack{\vec{p},\vec{v}_1,\ldots,\vec{v}_K}}{\text{maximize}}
		\; \min_{k\in \set{K}} \frac{\mathsf{SINR}_k^{\mathsf{inst}}(\vec{H},\vec{v}_k,\vec{p})}{\gamma_k}
		\text{ s.t. } (\forall k\in\mathcal{K})~p_k \leq P, 
\end{equation}
where $\mathcal{K}=\{1,\ldots,K\}$ is the set of user indexes, $(\gamma_1,\ldots,\gamma_K)\in \stdset{R}_{++}^K$ and $P\in \stdset{R}_+$ are problem parameters, and $\mathsf{SINR}_k^{\mathsf{inst}}$ denotes the instantaneous SINR of user $k$ for a given MU-MIMO channel matrix $\vec{H}\in \stdset{C}^{K\times M}$, power control coefficients $\vec{p}=(p_1,\ldots,p_K)\in \stdset{R}_+^K$, and receive beamformers $\vec{v}_1,\ldots,\vec{v}_K \in \stdset{C}^M$. The classical MU-MIMO literature provides many efficient and provably convergent algorithms for solving not only \eqref{eq:QoS_det} and \eqref{eq:maxmin_det}, but also many variants covering, for example, downlink transmission. Notably, these problems can be solved via fixed-point methods based on an axiomatic uplink power control framework introduced in \cite{yates95} for  sum power minimization problems, and extended in \cite{nuzman07} to weighted max-min problems. Furthermore, although developed for the uplink, this framework can also be applied to related downlink problems by means of uplink-downlink duality arguments \cite{schubert2004solution,yu2007transmitter,cai2011maxmin}. Embracing this abstract framework 
offers many advantages. In particular, it not only gives valuable insights into the network behavior at a high-level \cite{schubert2011interference,renato2016,renato2019}, but it also gives access to powerful numerical techniques for addressing more involved variants of \eqref{eq:QoS_det} and \eqref{eq:maxmin_det} without the need for ad-hoc proofs \cite{cai2011maxmin,cai2012comp,schubert2019multi,liu2019association}. 

In addition to sum power minimization and weighted max-min problems, another popular optimization criterion is the maximization of the weighted sum-rate subject to power constraints. However, solving weighted sum-rate maximization problems is notoriously challenging. In particular, in canonical MU-MIMO settings where interference is treated as noise, the sum-rate maximization problem is NP-hard in general \cite{luo2008dynamic} (taking aside restricted scenarios where specific symmetries or other channel characteristics simplify the problem). Therefore, either highly complex global optimization methods \cite{utschick2012monotonic,emil2013resource} or efficient suboptimal algorithms \cite{christensen2008wmmse,shen2018fractional} are typically considered in the literature. Since in this study we focus on efficient optimal designs, we consider only sum power minimization and weighted max-min problems, and we leave weighted sum-rate maximization problems for future work.

\subsection{Long-term joint power control and beamforming design}
The above conventional approaches, which we refer to as the \textit{short-term} approaches, suffer from the following  issues:
\begin{itemize}
\item First, \eqref{eq:QoS_det} and \eqref{eq:maxmin_det} require running an iterative algorithm \textit{for every  realization} of $\vec{H}$, which imposes demanding requirements in terms of computational complexity.
\item Second, \eqref{eq:QoS_det} and \eqref{eq:maxmin_det} require collecting the problem input $\vec{H}$ and distributing the solution $(\vec{p},\vec{v}_1,\ldots,\vec{v}_K)$ across all the appropriate network entities \textit{for every channel realization}, which results in high signaling overhead.
\end{itemize}
Importantly, note that all the above operations must be completed under strict latency constraints, since the optimized powers $\vec{p}$ must be conveyed to the users within each channel coherence block (in time and frequency). The main working assumption of this study, which is also ubiquitous in the massive MIMO literature (see, e.g., \cite{marzetta2016fundamentals,massivemimobook,ngo2017cell,nayebi2017precoding,bashar2019maxmin,bjornson2019making,du2021cellfree,demir2021,interdonato2021enhanced}  among others), is that the above challenges make short-term approaches impractical for large-scale MIMO systems, which calls for alternative approaches. 

Leaving heuristics aside, we can circumvent the above challenges by considering optimization problems involving ergodic rate expressions of the type 
\begin{equation}\label{eq:R}
\E\left[\log_2\left(1+\mathsf{SINR}_k^{\mathsf{inst}}(\rvec{H},\rvec{v}_k,\vec{p}\right)\right],
\end{equation}
where $\rvec{H}$ is a random channel matrix, $\vec{p}$ is a deterministic power vector, and where the beamformers $\rvec{v}_k$ are \textit{functions} of the available instantaneous channel state information (CSI). This two-timescale formulation, which we refer to as the \textit{long-term} approach, has the  advantage that the power vector $\vec{p}$ can be optimized only sporadically \textit{for many channel realizations} based on relatively slowly time-varying channel statistics. The same advantage applies to beamforming design, i.e., the optimization of the \emph{functions} mapping CSI to instantaneous beamforming coefficients. In addition, the long-term approach facilitates the introduction of non-trivial beamforming constraints modeling limited CSI sharing in distributed architectures, using the notion of information constraints proposed in \cite{miretti2021team}. Introducing these information constraints is particularly useful for optimizing cell-free networks in scenarios where beamformers must be computed locally at each access point (AP) based on local or partially shared CSI \cite{ngo2017cell,demir2021}. In particular, note that common AP clustering techniques \cite{buzzi2020,demir2021} alone are not sufficient to model these scenarios. However, the main drawback of the long-term approach based on \eqref{eq:R} is that it typically leads to intractable functional optimization problems. In fact, as of today, no optimal method for long-term joint power control and beamforming design is available in the literature. In this work, we address this limitation and present the first optimal yet tractable long-term method that overcomes the above challenges of short-term approaches.

\subsection{Main contribution and related long-term approaches}
Building on the related literature on long-term power control in (cell-free) massive MIMO systems \cite{marzetta2016fundamentals,massivemimobook,ngo2017cell,nayebi2017precoding,bashar2019maxmin,bjornson2019making,du2021cellfree,demir2021,interdonato2021enhanced}, the main contribution of this study is to show that, by replacing \eqref{eq:R} with the so-called \textit{use-and-then-forget} (UatF) bound on the achievable uplink ergodic rates \cite{marzetta2016fundamentals,massivemimobook}, we are able to formulate tractable two-timescale variants of \eqref{eq:QoS_det} and \eqref{eq:maxmin_det} to consider, for the first time, long-term power control and (possibly distributed) beamforming design. Specifically, we show that the proposed problem formulations enable us to derive globally optimal fixed-point methods based on the same axiomatic power control framework in \cite{yates95,nuzman07} used for solving \eqref{eq:QoS_det} and \eqref{eq:maxmin_det}. The key step of our derivation relies on a previously unexplored relation between the maximization of the UatF bound and a minimum mean-square error (MMSE) problem under arbitrary information constraints \cite{miretti2021team}. We remark that a similar relation is well-known for ergodic rate expressions of the type in \eqref{eq:R} \cite{massivemimobook}, but it does not lead to tractable joint long-term power control and beamforming design problems, and it is only valid for centralized beamforming  \cite{demir2021}.  

To relate the main contribution of the present study with existing long-term approaches, we recall that iterative algorithms for long-term power control under fixed beamforming design based on the UatF bound (or its related downlink version) have been extensively studied in both the cellular and cell-free massive MIMO literature (see, e.g., \cite{marzetta2016fundamentals,massivemimobook,ngo2017cell,nayebi2017precoding,du2021cellfree,bjornson2019making,demir2021,interdonato2021enhanced}). In particular, a variant of \eqref{eq:maxmin_det} based on the UatF bound, where only $\vec{p}$ is optimized, can be solved using the fixed-point method in \cite[Algorithm~7.1]{demir2021}. Furthermore, the above long-term power control problem has been generalized in \cite{bashar2019maxmin} to long-term power control and large-scale fading decoding (LSFD) design. Specifically, \cite{bashar2019maxmin} provides an iterative algorithm based on the UatF bound that not only optimizes $\vec{p}$, but also a set of statistical combining coefficients $\vec{a}_1,\ldots,\vec{a}_K\in \stdset{C}^M$ that are used to form enhanced beamfomers $\rvec{v}_k = \mathrm{diag}(\vec{a}_k)\rvec{v}_k^{\mathsf{fix}}$ from a fixed beamforming design $\rvec{v}_k^{\mathsf{fix}}$ (e.g., maximum-ratio combining as in \cite{bashar2019maxmin}). Our method generalizes (and, as we will see, outperforms) the long-term methods in \cite[Algorithm~7.1]{demir2021} and \cite{bashar2019maxmin} to long-term power control and beamforming design, i.e., to the joint optimization of $\vec{p}$ and of the functions $\rvec{v}_1,\ldots,\rvec{v}_K$ mapping instantaneous CSI to beamforming coefficients. In addition, we remark that a method for optimal beamforming design based on the UatF bound, which applies to centralized and also to distributed beamforming architectures, is provided in \cite{miretti2021team}, but the optimization of $\vec{p}$ is left as an open problem. Hence, our method also extends \cite{miretti2021team} by including optimal power control. In brief, the existing long-term optimization methods in the literature either perform optimal power control with suboptimal beamforming design, or optimize the beamformers under suboptimal power control. In contrast, our proposed technique solves the joint problem.

\subsection{Summary of contributions and organization of the paper}
In Section~\ref{sec:application} we formally define and motivate the considered long-term power control and beamforming design problems for large-scale MIMO systems. In  Section~\ref{sec:solution} we derive the proposed optimal algorithms for solving the above problems, by establishing a connection between the axiomatic power control framework in \cite{yates95,nuzman07} and the beamforming design method in \cite{miretti2021team}. For completeness, we provide in the appendix a concise yet up-to-date introduction to the framework in  \cite{yates95,nuzman07}, which uncovers overlooked connections within the existing literature. To illustrate the potential of the proposed approach, in Section~\ref{sec:sim} we conduct an extensive numerical study on canonical user-centric cell-free massive MIMO networks. We investigate many approaches for power control and beamforming design, and, as interesting byproducts of our main contribution, we obtain the following additional results:\footnote{A preliminary version of this study focusing on max-min fairness and distributed beamforming was presented in \cite{miretti2022joint}. This study extends \cite{miretti2022joint} by including the minimum sum power criterion and other beamforming models. Furthermore, it significantly improves the comparison against competing approaches, and it presents a more complete derivation and analysis of the algorithms based on an up-to-date review of fixed-point methods. In addition, it provides novel theoretical and experimental insights on cell-free networks.}
\begin{itemize}
\item We present the first optimal comparison of small-cells networks and cell-free massive MIMO networks regarding their ability to deliver uniform quality of service to all users. Specifically, our study takes into account optimal joint power control and potentially distributed beamforming design. Notably, we provide a straightforward explanation for performance differences using direct and intuitive arguments. Our approach departs from previous studies such as those in \cite{ngo2017cell,bjornson2019making,demir2021}, which, relying on suboptimal schemes, frequently yield results that are hard to interpret and may even lead to different conclusions. 
\item We contribute with novel insights to the ongoing debate on the appropriate ergodic rate bound for studying large-scale MIMO networks \cite{massivemimobook,caire2018ergodic,demir2021,gottsch2023subspace}. The UatF bound underestimates the true system performance \cite{caire2018ergodic}, but we give numerical evidence showing that it can serve as an excellent proxy for system optimization when dealing with intractable ergodic rate expressions such as those in \eqref{eq:R}. For the centralized beamforming case, this evidence is further strengthened by noticing that the known optimal beamforming design derived under an expression similar to \eqref{eq:R} \cite{massivemimobook,demir2021} coincides with the optimal design derived under the UatF bound.
\item We show that long-term power control and beamforming design may provide non-negligible performance gains compared to power control with fixed beamformers in both centralized and distributed cell-free massive MIMO networks. Furthermore, we show that competing short-term centralized methods may not only be impractical, but they also suffer from a non-negligible performance loss. Moreover, we confirm that beamfoming design is the main bottleneck in cell-free networks \cite{demir2021}.
\end{itemize}

We emphasize that this study focuses on comparing the solution to different problem \textit{formulations}, rather than different \textit{numerical methods} for a given problem formulation. In particular, our comparison does not cover different suboptimal methods, such as learning-based methods, that can be used to find less complex solutions to a given optimization problem \cite{sun2018learning,dandrea2019uplink,xia2020deep,garcia2024flexible,vahapoglu2024deep}. For instance, while a learning-based technique for solving \eqref{eq:QoS_det} and \eqref{eq:maxmin_det}, similar to those in \cite{xia2020deep,vahapoglu2024deep}, can be less complex than the algorithms proposed in \cite{rashid1998joint, schubert2004solution}, it would still face the same overhead challenges related to instantaneous uplink power control and centralized beamforming. Similarly, learning-based methods for long-term power control, such as those in \cite{dandrea2019uplink,garcia2024flexible}, do not overcome the limitations of suboptimal beamforming design. Note that, for the optimization problems we pose, learning-based methods can at best match the performance of the optimal approaches we derive. Nevertheless, we point out that exploring more efficient suboptimal techniques for the proposed two-timescale formulations represents an intriguing direction for future research.

\textit{Notation and mathematical preliminaries:}
Boldface lowercase letters, $\vec{x}$, denote column vectors, and boldface uppercase letters, $\vec{X}$, denote matrices. Random vectors and matrices are typographically distinguished from their deterministic counterparts using the blackboard typeface, $\rvec{x}$ and $\rvec{X}$. Calligraphic uppercase letters, $\set{X}$, denote sets, and the balckboard sans-serif symbols $\stdset{N}$, $\stdset{R}$, $\stdset{C}$ are reserved for the sets of natural, real, and complex numbers, respectively. The sets of nonnegative and positive reals are denoted by, respectively, $\real_+$ and $\real_{++}$. We denote by $x_k$ the $k$th coordinate of a vector $\vec{x}\in\real^K$. The $k$th column of the $K$-dimensional identity matrix $\vec{I}_K$ is denoted by $\vec{e}_k$.  The expectation and variance operators are denoted by $\E[\cdot]$ and $\Var(\cdot)$, respectively. Let $(\Omega,\Sigma,\mathbb{P})$ be a probability space. We denote by $\set{H}^K$ the set of complex valued random vectors, i.e., $K$-tuples of $\Sigma$-measurable functions $\Omega \to \stdset{C}$ satisfying $(\forall \rvec{x}\in \set{H}^K)$ $\E[\|\rvec{x}\|_2^2]<\infty$. Together with the standard operations of addition and real scalar multiplication, we recall that $\set{H}^K$ is a real vector space. Inequalities involving vectors should be understood coordinate-wise; i.e.,  $(\forall \vec{x}\in\real_+^K)(\forall \vec{y}\in\real_+^K)(\forall k\in\{1,\ldots,K\})~x_k\le y_k \Leftrightarrow \vec{x}\le\vec{y} $. A norm $\|\cdot\|$ in $\real^K$ is said to be monotone if $(\forall \vec{x}\in\real_+^K)(\forall \vec{y}\in\real_+^K)~ \vec{x}\le\vec{y}\Rightarrow \|\vec{x}\|\le\|\vec{y}\|$. We say that a sequence $(\vec{x}_n)_{n\in\Natural}\subset\real^K$ converges  to $\vec{x}^\star$ if $\lim_{n\to\infty}\|\vec{x}_n-\vec{x}^\star\|=0$ for some (and hence for every) norm $\|\cdot\|$ in $\real^K$. Given a mapping $T:\set{X}\to \set{Y}$, where $\set{Y} \subseteq \set{X}$, we denote by $\mathrm{Fix}(T):=\{ x \in \set{X}~|~T(x)=x\}$ its set of fixed points.  

\textit{Reproducible Research}: the simulation code is available at \url{https://github.com/lorenzomiretti/powercontrol}.

\section{Problem statement}
\label{sec:application}
We consider the uplink of a large-scale MIMO system composed by $K$ single-antenna users indexed by $\mathcal{K}:=\{1,\ldots,K\}$, and a radio access infrastructure equipped with $M$ antennas indexed by $\set{M}:=\{1,\ldots,M\}$. The $M$ infrastructure antennas may be split among many geographically distributed access points equipped with a relatively low number of antennas as in the cell-free massive MIMO concept \cite{ngo2017cell,demir2021}, or colocated in a single array as in the XL-MIMO concept \cite{debbah2023xl}. By assuming for each time-frequency resource a standard synchronous
narrowband MIMO channel model governed by a stationary ergodic fading process, and simple transmission techniques based on linear receiver processing and on treating interference as noise, we focus on simultaneously achievable ergodic rates in the classical Shannon sense, given by the popular \emph{use-and-then-forget} (UatF) inner bound on the ergodic capacity region \cite{marzetta2016fundamentals}. In more detail, we define the uplink rates achieved by each user $k \in \set{K}$ for a given power control policy and beamforming design as
\begin{equation}\label{eq:uatf}
	R_k(\rvec{v}_k,\vec{p}) \eqdef \log_2(1+\mathsf{SINR}_k(\rvec{v}_k,\vec{p})),
\end{equation}
\begin{equation*}
	\mathsf{SINR}_k(\rvec{v}_k,\vec{p}) \eqdef \resizebox{0.69\linewidth}{!}{$\dfrac{p_k|\E[\rvec{h}_k^\herm\rvec{v}_k]|^2}{p_k\Var(\rvec{h}_k^\herm\rvec{v}_k)+\underset{j\neq k}{\sum} p_j\E[|\rvec{h}_j^\herm\rvec{v}_k|^2]+\E[\|\rvec{v}_k\|_2^2]}$},
\end{equation*}
where $\vec{p} \eqdef (p_1,\ldots,p_K) \in \stdset{R}_{+}^K$ is a deterministic vector of transmit powers, $\rvec{h}_k \in\set{H}^M$ is a random channel vector modeling the fading state between user $k$ and all $M$ infrastructure antennas, and $\rvec{v}_k \in \set{H}^M$ models the beamforming vector which is applied by the infrastructure to process the received signals of potentially all $M$ antennas to obtain a soft estimate of the transmit signal of user~$k$. We stress that $\rvec{v}_k$ is generally a function of the instantaneous CSI realizations, and hence it is denoted as a random vector. However, $R_k(\rvec{v}_k,\vec{p})$ is not a random variable because of the expectation operator, which is taken with respect to all involved random variables (as in every other part of this manuscript). For convenience, we also define the global channel matrix $\rvec{H}\eqdef[\rvec{h}_1,\ldots,\rvec{h}_K]$.

\begin{remark}
\label{rem:bounds} 
The rate expression given by \eqref{eq:uatf} assumes that CSI is neglected in the decoding phase. Alternative rate expressions of the type in \eqref{eq:R} with various forms of CSI at the decoder are often considered in the literature (see \cite{massivemimobook,caire2018ergodic,demir2021,gottsch2023subspace} and references therein). These expressions may provide less pessimistic capacity approximations, but system optimization often becomes challenging. In contrast, in this study we show that \eqref{eq:uatf} leads to tractable problem formulations. In addition, in Section~\ref{sec:sim} we also give numerical evidence showing that \eqref{eq:uatf} can serve as an excellent proxy for optimization according to intractable ergodic rate expressions such as those in \eqref{eq:R}.
\end{remark}

The classical MU-MIMO paradigm typically assumes the capability of acquiring high-quality measurements of $\rvec{H}$ by means of pilot-based channel estimation techniques, such as those described in \cite{massivemimobook}, to jointly process the signals of all antennas in real-time. However, very large-scale systems are often impaired by practical limitations in processing power and in the interconnections among antennas and processing units, which makes the above paradigm highly impractical. Therefore, following the modern literature on large-scale MIMO systems \cite{ngo2017cell,demir2021,debbah2023xl}, in this work we assume that the system operates under practical constraints that limit the CSI acquisition and joint processing burden, as described in the next sections.

\subsection{Long-term joint power control and beamforming design}
Based on the above model,  our objective is to solve the following two optimization problems:
\begin{equation}\label{eq:QoS}
	\begin{aligned}
		\underset{\substack{\vec{p}\in \stdset{R}_{+}^K \\ \rvec{v}_1,\ldots,\rvec{v}_K \in \set{H}^{M} }}{\text{minimize}}
		& \quad  \|\vec{p}\|_1 \\
		\text{subject to} & \quad (\forall k \in \set{K})~\mathsf{SINR}(\rvec{v}_k,\vec{p}) \geq \gamma_k\\
		& \quad  (\forall k \in \set{K})~\rvec{v}_k \in \set{V}_k\\
		& \quad  (\forall k \in \set{K})~\E[\|\rvec{v}_k\|_2^2] \neq 0; 
	\end{aligned}
\end{equation}
and
\begin{equation}\label{eq:maxmin}
	\begin{aligned}
		\underset{\substack{\vec{p}\in \stdset{R}_{+}^K \\ \rvec{v}_1,\ldots,\rvec{v}_K \in \set{H}^{M} }}{\text{maximize}}
		& \quad  \min_{k\in\set{K}} \gamma_k^{-1} \mathsf{SINR}_k(\rvec{v}_k,\vec{p}) \\
		\text{subject to} & \quad \|\vec{p}\|_{\infty}\leq P \\
		& \quad  (\forall k \in \set{K})~\rvec{v}_k \in \set{V}_k\\
		& \quad  (\forall k \in \set{K})~\E[\|\rvec{v}_k\|_2^2] \neq 0,
	\end{aligned}
\end{equation}
where $(\gamma_1,\ldots,\gamma_K)\in \stdset{R}_{++}^K$ is a given vector of positive coefficients denoting either minimum SINR requirements or user weights, $P \in \stdset{R}_{++}$ is a given power budget, and the sets $\set{V}_1,\ldots,\set{V}_K \subseteq \set{H}^M$ are given \textit{information constraints} (formally defined in Section~\ref{ssec:info}) that are imposed to limit the joint processing and information sharing burden, thus enabling scalable designs for large-scale MIMO systems.  For example, as we will see, the set $\set{V}_k$ can be used to limit the dependency of the entries of $\rvec{v}_k$ to portions of the CSI available at the infrastructure side, and also the number of antennas that are serving user $k$. Note that letting $\set{V}_k = \set{H}^M$ (i.e., omitting the information constraint) allows $\rvec{v}_k$ to be essentially any function of the global channel $\rvec{H}$, which is not practical. However, we remark that a solution to \eqref{eq:QoS} or \eqref{eq:maxmin} is generally not known even by omitting the information constraints.

\begin{remark}\label{rem:competing} The long-term power control algorithm in \cite[Algorithm~7.1]{demir2021} solves a particular case of \eqref{eq:maxmin} recovered by letting $\set{V}_k = \{\rvec{v}_k^{\mathsf{fix}}\}$ for some fixed $\rvec{v}^{\mathsf{fix}}_k\in \set{H}^M$, i.e., by restricting the set of feasible beamformers to a singleton (e.g., $\set{V}_k = \{\rvec{h}_k\}$ for the case of maximum-ratio combining with perfect CSI). Similarly, the long-term joint power control and LSFD design algorithm in \cite{bashar2019maxmin} solves a particular case of \eqref{eq:maxmin} recovered by letting $\set{V}_k = \{\mathrm{diag}(\vec{a}_k)\rvec{v}_k^{\mathsf{fix}}~|~\vec{a}_k \in \stdset{C}^M\}$ for some $\rvec{v}_k^{\mathsf{fix}}\in \set{H}^M$. Additional comparisons are given in Section~\ref{sec:sim}.
\end{remark}

The proposed problems share the following features:
\begin{itemize}
	\item The power control vector $\vec{p}$ is kept fixed for the whole duration of the transmission, i.e., it does not depend on the small-scale fading realizations, and it is optimized based on channel statistics. 
	\item The beamformers $(\rvec{v}_1,\ldots,\rvec{v}_K)$ adapt to small-scale fading based on the CSI specified by the information constraints. However, their design, i.e., the function mapping CSI realizations to the corresponding beamforming coefficients, is optimized based on channel statistics as for the power control vector. Note that this leads to functional optimization problems in infinite-dimensional vector spaces, for which standard solvers in finite-dimensional Euclidean spaces (e.g., \texttt{cvx}) are not applicable. 
\end{itemize}
Compared to traditional short-term approaches that solve an optimization problem such as \eqref{eq:QoS_det} and \eqref{eq:maxmin_det} for each channel (or global CSI) realization, the long-term approach in this study is particularly attractive for large-scale MIMO systems due to the following reasons:
\begin{itemize}
	\item First, the power control vector and beamforming design can be updated following large-scale fading variations, which are typically several order of magnitude slower than small-scale fading variations.
	\item Second, the beamforming design space covers distributed implementations where the computation of the beamformers can be split among multiple distributed processing units with limited CSI sharing capabilities, each controlling a subset of the $M$ infrastructure antennas.     
\end{itemize}

\begin{remark}
A known challenge in long-term approaches is the need to identify and track some form of channel statistics that can be used to produce effective solutions. In contrast, short-term approaches can be often implemented based on minimal or even no statistical information. However, the main working assumption of this study is that short-term approaches are not feasible in large-scale MIMO systems, thus making long-term approaches not only an option but rather a necessity.
\end{remark}
	
\subsection{Information constraints}
\label{ssec:info}
In distributed beamforming implementations, each processing unit must often form its portion of the beamformers on the basis of different CSI. In particular, the distributed processing units typically acquire CSI from locally received pilot signals, and then further share these measurements across the network. However, due to limited interconnections, these measurements may not be perfectly shared. As an extreme example, in distributed cell-free massive MIMO networks the access points are typically assumed to compute the beamformers on the basis of only \emph{local} CSI \cite{ngo2017cell}, i.e., without any CSI sharing. However, to date, most related studies do not formally model this type of constraint, so, strictly speaking, they do not truly consider distributed beamforming optimization. To address this limitation, we model constraints related to imperfect CSI sharing following the approach introduced in \cite{miretti2021team}. We first partition the $M = NL$ infrastructure antennas into $L$ groups of $N$ antennas, and assume that the corresponding portions of the beamformers are computed by $L$ separate processing units. Given an arbitrarily distributed tuple $(\rvec{H},S_1,\ldots,S_L)$, where $S_l$ is the available CSI at the $l$th processing unit, we then let 
	\begin{equation}\label{eq:distributedCSI}
		\set{V}_k \eqdef \set{H}_1^N \times \ldots \times \set{H}_L^N,
	\end{equation}
	where $\set{H}_l^N\subseteq \set{H}^N$ denotes the set of $N$-tuples of $\Sigma_l$-measurable functions $\Omega \to \stdset{C}^N$ satisfying $(\forall \rvec{x}\in \set{H}_l^N)$ $\E[\|\rvec{x}\|_2^2]<\infty$, and where $\Sigma_l \subseteq \Sigma$ is the sub-$\sigma$-algebra induced by the CSI $S_l$, which is also called the \emph{information subfield} of the $l$th processing unit. The interested reader is referred to \cite{yukselbook} for an introduction to the measure theoretical notions used in the above definitions. However, we stress that these notions are by no means required for understanding the key results of this study. The crucial point is that, informally, the constraint set $\set{V}_k$ enforces the portion of $\rvec{v}_k$ computed by the $l$th processing unit to be a function of only $S_l$. 

Furthermore, another typical constraint for enabling scalable large-scale MIMO designs is that only a subset of all infrastructure antennas should be used for decoding the message of a given user. For example, cell-free massive MIMO networks are expected to implement some type of user-centric rule \cite{buzzi2020}, which allocates to each user a cluster of serving access points. As shown in \cite{miretti2021team2,miretti2023duality}, this constraint can be modeled by forcing the corresponding portions of $\rvec{v}_k$ to zero, i.e., by replacing $\set{H}_l^N$ in \eqref{eq:distributedCSI} with the set $\{\vec{0}_N~ (\text{a.s.})\}$ if the $l$th processing unit is not serving the $k$th user. Considering this replacement, we remark that $\set{V}_k$ is a subspace of the real vector space $\set{H}^M$ \cite{miretti2023duality}. Concrete examples are provided in Section~\ref{sec:sim}.
	
\section{Problem solution}
\label{sec:solution}
We now derive novel optimal algorithms for solving both Problem~\eqref{eq:QoS} and Problem~\eqref{eq:maxmin}. Our derivation combines the abstract power control framework in \cite{yates95,nuzman07} with the general beamforming design framework in \cite{miretti2021team} (which applies to both centralized and distributed beamforming architectures). Our algorithms are a nontrivial extension of the available algorithms based on \cite{yates95,nuzman07} for solving the related short-term problems \eqref{eq:QoS_det} and \eqref{eq:maxmin_det}.
 	\subsection{Power control with implicit beamforming optimization}
We now map Problem~\eqref{eq:QoS} and Problem~\eqref{eq:maxmin} to equivalent power control problems where the optimization of the beamformers is implicit. More precisely, by exploiting the fact that uplink SINRs are only coupled through the power vector $\vec{p}$, we consider the following optimization problems:
\begin{equation}\label{eq:yates}
	\begin{aligned}
		\underset{\substack{\vec{p}\in \stdset{R}_{+}^K }}{\text{minimize}}
		& \quad  \|\vec{p}\| \\
		\text{subject to} & \quad (\forall k\in\mathcal{K})~u_k(\vec{p}) \geq \gamma_k;
	\end{aligned}
\end{equation}
and
\begin{equation}\label{eq:nuzman}
\begin{aligned}
		\underset{\substack{\vec{p}\in \stdset{R}_{+}^K}}{\text{maximize}}
		& \quad  \min_{k\in\set{K}} \gamma_k^{-1} u_k(\vec{p}) \\
		\text{subject to} & \quad \|\vec{p}\|\leq P,
	\end{aligned}
\end{equation}
where $\|\cdot\|$ denotes a given monotone norm, and where, for each user $k\in \set{K}$, we define the utility
\begin{equation}\label{eq:maxSINR}
(\forall\vec{p}\in\stdset{R}_+^K)~u_k(\vec{p}) \eqdef \sup_{\substack{\rvec{v}_k\in\set{V}_k\\ \E[\|\rvec{v}_k\|_2^2] \neq 0}}\mathsf{SINR}_k(\rvec{v}_k,\vec{p}).
\end{equation}
The above problems can be formally connected to the original long-term joint power control and beamforming design problems as follows.
\begin{proposition}\label{prop:connection}
Let $\vec{p}^\star$ be a solution to Problem~\eqref{eq:yates} (resp. Problem~\eqref{eq:nuzman}) with utilities in \eqref{eq:maxSINR} and norm $\|\cdot\| = \|\cdot\|_1$ (resp., $\|\cdot\| = \|\cdot\|_\infty$). If $(\forall k~\in \set{K})$ $\exists \rvec{v}^\star_k\in \set{V}_k$ such that $u_k(\vec{p}^\star) = \mathsf{SINR}_k(\rvec{v}^\star_k,\vec{p}^\star)$ (i.e., the supremum in \eqref{eq:maxSINR} is attained), then $(\vec{p}^\star,\rvec{v}_1^\star,\ldots,\rvec{v}_K^\star)$ solves Problem~\eqref{eq:QoS} (resp. Problem~\eqref{eq:maxmin}).
\end{proposition}
\begin{proof}
Let us focus on the connection between \eqref{eq:QoS} and \eqref{eq:yates}. By construction, $(\vec{p}^\star,\rvec{v}_1^\star,\ldots,\rvec{v}_K^\star)$ satisfies all constraints of Problem~(5). Hence, $\|\vec{p}^\star\|_1 \geq \tilde{p}$, where $\tilde{p}$ denotes the optimum of Problem~(5). Furthermore, let $(\vec{p},\rvec{v}_1,\ldots,\rvec{v}_K)$ satisfy all constraints of Problem~(5). Then, by definition of supremum, we have that $(\forall k \in \set{K})~u_k(\vec{p}) \geq \mathsf{SINR}_k(\rvec{v}_k,\vec{p}) \geq \gamma_k$, i.e., $\vec{p}$ also satisfies all constraints of Problem~(8). This in turns implies that $\tilde{p}\geq \|\vec{p}^\star\|_1$. By combining the above arguments, we conclude that $\tilde{p}= \|\vec{p}^\star\|_1$ and that  $(\vec{p}^\star,\rvec{v}_1^\star,\ldots,\rvec{v}_K^\star)$ solves Problem~(5). The connection  between \eqref{eq:maxmin} and \eqref{eq:nuzman} follows from similar arguments.

By definition, $(\forall k \in \set{K})~(\forall\vec{p}\in\stdset{R}_+^K)$ $u_k(\vec{p})\geq \mathsf{SINR}_k(\rvec{v}_k,\vec{p})$ for all $\rvec{v}_k \in \set{V}_k$ satisfying $\E[\|\rvec{v}_k\|_2^2] \neq 0$. Furthermore, $\mathsf{SINR}_k(\rvec{v}_k,\vec{p})$ does not depend on the beamformer $\rvec{v}_j$ of any other user $j \neq k$. These two properties imply that the optimum of Problem~\eqref{eq:yates} (resp. Problem~\eqref{eq:nuzman}) with utilities in \eqref{eq:maxSINR} and norm $\|\cdot\| = \|\cdot\|_1$ (resp., $\|\cdot\| = \|\cdot\|_\infty$) is also the optimum of Problem~\eqref{eq:QoS} (resp. Problem~\eqref{eq:maxmin}). This, combined with the considered assumptions on $(\vec{p}^\star,\rvec{v}_1^\star,\ldots,\rvec{v}_K^\star)$, further implies the main statement of the proposition.
\end{proof}

We now demonstrate a key structural property of the utility in \eqref{eq:maxSINR}, which allows us to connect Problem~\eqref{eq:yates} and Problem~\eqref{eq:nuzman} to the axiomatic power control framework in \cite{yates95,nuzman07}, reviewed in the appendix. 
\begin{proposition}\label{prop:fk}
Assume the technical condition $(\forall k~\in \set{K})$ $\set{V}'_k \eqdef \{\rvec{v}_k \in \set{V}_k~|~\E[\rvec{h}_k^\herm \rvec{v}_k]\neq 0\} \neq \emptyset$. Then, for all $k \in \set{K}$ and $\vec{p}\in\stdset{R}_+^K$, the utility \eqref{eq:maxSINR} takes the form $u_k(\vec{p})=\frac{p_k}{f_k(\vec{p})}$, where
\begin{equation*}
~f_k(\vec{p}) \eqdef\inf_{\rvec{v}_k \in \set{V}'_k} \dfrac{p_k\Var(\rvec{h}_k^\herm\rvec{v}_k)+\underset{j\neq k}{\sum}p_j\E[|\rvec{h}_j^\herm\rvec{v}_k|^2]+\E[\|\rvec{v}_k\|_2^2]}{|\E[\rvec{h}_k^\herm\rvec{v}_k]|^2}
\end{equation*} is a positive concave function.
\end{proposition}
\begin{proof}
The following steps hold for all $k \in \set{K}$ and $\vec{p}\in\stdset{R}_{+}^K$. By the Cauchy-Schwarz inequality, we first observe that the bound $\mathsf{SINR}_k(\rvec{v}_k,\vec{p}) \leq \frac{p_k|\E[\rvec{h}_k^\herm\rvec{v}_k]|^2}{\E[\|\rvec{v}_k\|_2^2]} \leq \frac{p_k\E[\|\rvec{h}_k\|_2^2]\E[\|\rvec{v}_k\|_2^2]}{\E[\|\rvec{v}_k\|_2^2]}=p_k\E[\|\rvec{h}_k\|_2^2]$ holds for all $\rvec{v}_k \in \set{V}_k$ such that $\E[\|\rvec{v}_k\|_2^2]\neq 0$. Then, by noticing that $\rvec{v}_k \in \set{V}'_k \implies \E[\|\rvec{v}_k\|_2^2]\neq 0$, and hence that $\emptyset \neq \set{V}'_k  \subset \{\rvec{v}_k \in \set{V}_k~|~\E[\|\rvec{v}_k\|_2^2] \neq 0\}$, we obtain 
\begin{align*}
&u_k(\vec{p})= \sup_{\substack{\rvec{v}_k\in\set{V}_k\\ \E[\|\rvec{v}_k\|_2^2] \neq 0}}\mathsf{SINR}_k(\rvec{v}_k,\vec{p}) \\
&= \max\Big\{\sup_{\rvec{v}_k\in\set{V}'_k}\mathsf{SINR}_k(\rvec{v}_k,\vec{p}), \sup_{\substack{\rvec{v}_k\in\set{V}_k\backslash \set{V}_k'\\ \E[\|\rvec{v}_k\|_2^2] \neq 0}}\mathsf{SINR}_k(\rvec{v}_k,\vec{p}) \Big\} \\
&= \max\Big\{\sup_{\rvec{v}_k\in\set{V}'_k}\mathsf{SINR}_k(\rvec{v}_k,\vec{p}),0 \Big\} = \sup_{\rvec{v}_k\in\set{V}'_k}\mathsf{SINR}_k(\rvec{v}_k,\vec{p})\\
&= p_k \sup_{\rvec{v}_k\in\set{V}'_k} \dfrac{|\E[\rvec{h}_k^\herm\rvec{v}_k]|^2}{p_k\Var(\rvec{h}_k^\herm\rvec{v}_k)+\underset{j\neq k}{\sum}p_j\E[|\rvec{h}_j^\herm\rvec{v}_k|^2]+\E[\|\rvec{v}_k\|_2^2]}\\
&= p_k/f_k(\vec{p}),
\end{align*}
where the second equality follows from elementary properties on the supremum of the union of bounded sets, and where the third  equality follows by definition of $\set{V}_k'$. Furthermore, we observe that $f_k$ is concave because it is constructed by taking the point-wise infimum of affine (and hence concave) functions. Finally, again by the Cauchy-Schwarz inequality, we observe that $f_k(\vec{p})\geq \inf_{\rvec{v}_k \in \set{V}'_k}\frac{\E[\|\rvec{v}_k\|_2^2]}{|\E[\rvec{h}_k^\herm\rvec{v}_k]|^2} \geq \frac{1}{\E[\|\rvec{h}_k\|_2^2]} >0$, where the last inequality follows since $\set{V}_k' \neq \emptyset \implies \E[\|\rvec{h}_k\|_2^2] \neq 0$. 
\end{proof}
In particular, the above proposition allows us to characterize and compute a solution to Problem~\eqref{eq:yates} and Problem~\eqref{eq:nuzman} by focusing on the mapping 
\begin{equation*}
(\forall \vec{p}\in \stdset{R}_{+}^K)~T(\vec{p})\eqdef (\gamma_1f_1(\vec{p}),\ldots,\gamma_Kf_K(\vec{p})),
\end{equation*}
as stated in the following two propositions.

\begin{proposition}\label{prop:optimum_yates} Assume $(\forall k~\in \set{K})$ $\set{V}'_k  \neq \emptyset$, and that a solution to Problem~\eqref{eq:yates} exists. Then, each of the following holds:
\begin{enumerate}[(i)]
		\item The set $\mathrm{Fix}(T)$ of fixed points of the mapping $T$ is a singleton, and the unique fixed point $\vec{p}^\star\in\mathrm{Fix}(T)$ is the unique solution to Problem~\eqref{eq:yates}.
		\item For every $\vec{p}_1 \in \stdset{R}_{+}^K$, the sequence $(\vec{p}_i)_{i\in\stdset{N}}$ generated via $(\forall i \in \stdset{N})~\vec{p}_{i+1} = T(\vec{p}_i)$ converges to $\vec{p}^\star\in\mathrm{Fix}(T)$.
\end{enumerate}
\end{proposition}
\begin{proof}
By Proposition~\ref{prop:fk}, the utilities in \eqref{eq:maxSINR} satisfy the axioms of the power control framework in \cite{yates95} (see  Proposition~\ref{proposition.concavity} in the appendix). Hence, the proof follows from a direct application of the results in \cite{yates95}, reviewed in Appendix~\ref{app:yates}.
\end{proof}

\begin{proposition}\label{prop:optimum_nuzman} Assume $(\forall k~\in \set{K})$ $\set{V}'_k  \neq \emptyset$. Then, each of the following holds:
\begin{enumerate}[(i)]
		\item The set $\mathrm{Fix}(\tilde{T})$ of fixed points of the mapping $(\forall \vec{p}\in \stdset{R}_+^K)$ $\tilde{T}(\vec{p})\eqdef \frac{P}{\|T(\vec{p})\|}T(\vec{p})$ is a singleton, and the unique fixed point $\vec{p}^\star\in\mathrm{Fix}(\tilde{T})$  is a solution to  Problem~\eqref{eq:nuzman}.
\item For every $\vec{p}_1 \in \stdset{R}_{+}^K$, the sequence $(\vec{p}_i)_{i\in\stdset{N}}$ generated via $(\forall i \in \stdset{N})~\vec{p}_{i+1} = \tilde{T}(\vec{p}_i)$ converges to $\vec{p}^\star\in\mathrm{Fix}(\tilde{T})$.
\end{enumerate}
\end{proposition}
\begin{proof}
Similarly to Proposition~\ref{prop:optimum_yates}, the proof follows from an application of the results in \cite{nuzman07}, reviewed in Appendix~\ref{app:nuzman}. 
\end{proof}
The main implication of Proposition~\ref{prop:optimum_yates} and Proposition~\ref{prop:optimum_nuzman} is that, if the mild technical condition $(\forall k \in \set{K})$ $\set{V}_k' \neq \emptyset$ holds, an optimal power vector $\vec{p}^\star$ for Problem~\eqref{eq:QoS} and Problem~\eqref{eq:maxmin} can be obtained via (possibly normalized) fixed-point iterations, where each $i$th step involves the evaluation of the mapping
\begin{equation*}
T(\vec{p}) = \left(\frac{\gamma_1p_1}{u_1(\vec{p})},\ldots,\frac{\gamma_Kp_K}{u_K(\vec{p})} \right)
\end{equation*}
for $\vec{p}=\vec{p}_i\in \stdset{R}_{++}^K$, which in turn involves the evaluation of $u_k(\vec{p})$ in \eqref{eq:maxSINR}. The main challenge of this approach lies in the computation of $u_k(\vec{p})$, which requires solving a SINR maximization problem over the set of feasible beamformers. While this is trivial for the two particular cases from the previous literature described in Remark~\ref{rem:competing} (supremum of a singleton, and maximization of a generalized Rayleigh quotient, respectively), the more general setup considered in this work requires additional care, even for the relatively simple case of centralized beamforming design. The following section discusses how to perform this step, and it shows that the supremum in \eqref{eq:maxSINR} is attained for all $\vec{p}\in \stdset{R}_{++}^K$ under the technical condition $\set{V}_k' \neq \emptyset$, which is typically satisfied in practice. (Note that $\set{V}_k' = \emptyset$ leads to trivial problem formulations, as it forces the SINR to zero.)  
	
\begin{remark} The condition $\set{V}_k'\neq \emptyset$ holds for most fading and CSI distributions of practical interest. However, note that $\set{V}_k'\neq \emptyset$ does not hold if the infrastructure has no CSI (equivalently, if $\rvec{v}_k$ and $\rvec{H}$ are independent) and, in addition, if the channels have zero mean. In this case, $(\forall \rvec{v}_k \in \set{V}_k)~\E[\rvec{h}_k^\herm\rvec{v}_k] = \E[\rvec{h}_k]^\herm\E[\rvec{v}_k] = 0$ holds. This is related to one of the main drawbacks of the UatF bounding technique, which requires either sufficient CSI or a strong channel mean to produce nonzero rate bounds.
\end{remark}

\subsection{SINR maximization via MSE minimization}
Solving the SINR maximization problems to compute the functions $u_k$ in \eqref{eq:maxSINR} appears quite challenging for the following reasons: (i) the non-convex fractional objective involving expectations in both the numerator and denominator, and (ii) the information constraints. However, the next proposition shows that (i) is not the main challenge, because we can consider an alternative and simpler convex utility, in the same spirit of the known relation between SINR and MSE for fixed channel matrices. Let us consider for all $k\in \set{K}$ the MSE between the signal $x_k\sim \CN(0,1)$ of user~$k$ and its soft estimate $\hat{x}_k \eqdef \rvec{v}_k^\herm\rvec{y}$ obtained by processing the output of the uplink channel $\rvec{y} \eqdef \sum_{k\in \set{K}}\sqrt{p_k}\rvec{h}_kx_k + \rvec{n}$ with noise $\rvec{n}\sim \CN(\vec{0},\vec{I})$, where the noise and all user signals are mutually independent and independent of $(\rvec{H},\rvec{v}_1,\ldots,\rvec{v}_K)$. Specifically, let $(\forall k \in \set{K})(\forall \rvec{v}_k \in \set{V}_k)(\forall \vec{p}\in \stdset{R}_+^K)$
	\begin{equation*}
		\begin{split}
			\mathsf{MSE}_k(\rvec{v}_k,\vec{p}) \eqdef & \; \E[|x_k - \hat{x}_k|^2]\\
			=& \; \E\left[\|\vec{P}^{\frac{1}{2}}\rvec{H}^\herm\rvec{v}_k-\vec{e}_k\|_2^2\right] + \E\left[\|\rvec{v}_k\|_2^2\right],
		\end{split}
	\end{equation*} 
	where $\vec{P}\eqdef \mathrm{diag}(\vec{p})$, and where the second equality can be verified via simple manipulations. We then have:
	\begin{proposition}
		\label{prop:MSE}
		For given $k\in\set{K}$ and $\vec{p}\in\stdset{R}_{+}^K$, consider the optimization problem 
		\begin{equation}\label{eq:MSE}
			\underset{\rvec{v}_k \in \set{V}_k}{\emph{minimize}}~\mathsf{MSE}_k(\rvec{v}_k,\vec{p}).
		\end{equation}
		Problem~\eqref{eq:MSE} has a unique solution $\rvec{v}_k^\star\in \set{V}_k$. Furthermore,  
\begin{equation*}
\mathsf{MSE}_k(\rvec{v}_k^\star,\vec{p}) = \frac{1}{1+u_k(\vec{p})}.
\end{equation*} 
Moreover, if $\set{V}_k'\neq \emptyset$ and $p_k >0$, then $
			u_k(\vec{p}) = \mathsf{SINR}_k\left(\rvec{v}_k^\star,\vec{p}\right)$.
	\end{proposition}
	\begin{proof} 
The proof follows similar steps as in the proof of \cite[Proposition~8]{miretti2023duality}, and they are reported here for better clarity. Existence and uniqueness of the solution follows from \cite[Lemma~4]{miretti2021team}. For the second statement, we observe that  
\begin{align*}
&\mathsf{MSE}_k(\rvec{v}^\star_k,\vec{p})=\inf_{\rvec{v}_k\in\mathcal{V}_k}\mathsf{MSE}_k(\rvec{v}_k,\vec{p}) \\
			&\overset{(a)}{=} \inf_{\substack{\rvec{v}_k\in\set{V}_k\\ \E[\|\rvec{v}_k\|_2^2] \neq 0}} \inf_{\beta \in \stdset{C}} \mathsf{MSE}_k(\beta\rvec{v}_k,\vec{p}) \\
			&\overset{(b)}{=} \inf_{\substack{\rvec{v}_k\in\set{V}_k\\ \E[\|\rvec{v}_k\|_2^2] \neq 0}} 1-\dfrac{p_k|\E[\rvec{h}_k^\herm\rvec{v}_k]|^2}{\sum_{j\in\set{K}}p_j\E[|\rvec{h}_j^\herm\rvec{v}_k|^2]+\E\left[\|\rvec{v}_k\|_2^2\right]} \\
			&=  \inf_{\substack{\rvec{v}_k\in\set{V}_k\\ \E[\|\rvec{v}_k\|_2^2] \neq 0}}\dfrac{1}{1+\mathsf{SINR}_k(\rvec{v}_k,\vec{p})} = \dfrac{1}{1+u_k(\vec{p})},
\end{align*}
where $(a)$ follows from $(\forall \beta\in \stdset{C})(\forall \rvec{v}_k \in \mathcal{V}_k)$ $\beta \rvec{v}_k \in\mathcal{V}_k$, and where $(b)$ follows from standard minimization of scalar quadratic forms (see, e.g., \cite[Appendix~D]{miretti2023duality}). 
It remains to verify that if $\set{V}_k'\neq \emptyset$ and $p_k>0$, then $\rvec{v}_k^\star$ also attains the supremum in \eqref{eq:maxSINR}. This can be verified by first observing that, by Proposition~\ref{prop:fk},  $u_k(\vec{p})>0 \implies \mathsf{MSE}_k(\rvec{v}^\star_k,\vec{p}) < 1 \implies \E[\|\rvec{v}_k^\star\|_2^2] \neq 0$, where the last step follows since $\E[\|\rvec{v}_k^\star\|_2^2] = 0 \iff \rvec{v}_k^\star = \vec{0} \text{ a.s.}\implies \mathsf{MSE}_k(\rvec{v}^\star_k,\vec{p}) = 1$. Then, the proof is concluded by following similar steps as in the above chain of equalities: $
\mathsf{MSE}_k(\rvec{v}_k^\star,\vec{p}) =  \inf_{\beta \in \stdset{C}} \mathsf{MSE}_k(\beta\rvec{v}_k^\star,\vec{p}) = (1+\mathsf{SINR}_k(\rvec{v}_k^\star,\vec{p}))^{-1}$.  
\end{proof}
	In the particular case where all users are served by all antennas based on perfect knowledge of the channel $\rvec{H}$ (trivial information constraints), the solution to Problem~\eqref{eq:MSE} is simply given by the well-known MMSE beamforming solution
	\begin{equation}\label{eq:MMSE}
		(\forall k \in \set{K})~\rvec{v}_k = \left(\rvec{H}\vec{P}\rvec{H}^\herm + \vec{I}\right)^{-1}\rvec{h}_k\sqrt{p_k}.
	\end{equation}
	Variations of \eqref{eq:MMSE} are readily obtained also for the case of imperfect yet perfectly shared CSI among all processing units serving each user, under mild assumptions on the channel estimation error model (see, e.g., \cite[Corollary~4.3]{massivemimobook} or the example in Section~\ref{sec:sim}).
	
	For more general and nontrivial information constraints, the solution to Problem~\eqref{eq:MSE} can be interpreted as the best distributed approximation of regularized channel inversion, and it can be obtained via the recently developed \emph{team} MMSE beamforming method given in \cite{miretti2021team}.\footnote{Although presented in the context of downlink beamforming, the method in \cite{miretti2021team} directly applies to the uplink case.} This method states that the solution to \eqref{eq:MSE} corresponds to the unique solution to a linear feasibility problem, which, depending on the information constraints, can be solved in closed-form or approximately via efficient numerical techniques. For example, in the context of user-centric cell-free massive MIMO networks, this method is used in \cite{miretti2021team} \cite{miretti2023duality} to provide the first explicit solution to the problem of optimal local beamforming design. Furthermore, it can be applied to more advanced information constraints involving delayed CSI sharing \cite{miretti2024delayed}, or exploiting the peculiarities of serial fronthauls such as in the so-called \emph{radio stripe} concept \cite{miretti2021team} \cite{miretti2021team2}. 
	
\begin{remark}\label{rem:longterm}
The connection between the MMSE solution in \eqref{eq:MMSE} (or its variations) and more traditional uplink ergodic rate bounds of the type in \eqref{eq:R}, is well-known (see, e.g., \cite[Corollary~4.2]{massivemimobook}). However, the connection between  Problem~\eqref{eq:MSE} and the UatF bound in \eqref{eq:uatf} given by Proposition~\ref{prop:MSE} seems largely overlooked in the literature. A known advantage of \eqref{eq:uatf} is that it makes long-term power control tractable \cite{marzetta2016fundamentals,massivemimobook}. This study shows that this tractability also extends to joint long-term power control and beamforming design, and it even allows to cover distributed beamforming design under nontrivial information constraints. In addition, \eqref{eq:uatf} can be rigorously connected to a similar downlink ergodic rate bound (the so-called hardening bound) by means of uplink-downlink duality arguments under sum \cite{massivemimobook} and per-antenna power constraints \cite{miretti2023duality}, hence providing optimal downlink designs.
\end{remark}
	
\subsection{Proposed algorithms}
We now have all the elements for deriving an algorithmic solution to both Problem~\eqref{eq:QoS} and Problem~\eqref{eq:maxmin}, i.e., to perform joint optimal long-term power control and beamforming design according to the considered criteria. To keep the proposed algorithms general, in the following section we do not specify the setup (i.e., $\rvec{H}$ and the information constraints). For particular applications to cell-free massive MIMO networks, we refer to our numerical examples in Section~\ref{sec:sim}. Specifically, we propose to solve Problem~\eqref{eq:QoS} via the fixed-point iterations in Proposition~\ref{prop:optimum_yates}(ii) combined with the general beamforming design criterion in Problem~\eqref{eq:MSE}, as summarized below:
	
\begin{center}	
\begin{tabular}{|c|}
\hline
\textbf{Algorithm 1} for solving Problem~\eqref{eq:QoS} \\
\hline
\begin{minipage}{0.8\linewidth}
	\begin{algorithmic}[1]
		\Require $\vec{p} \in \stdset{R}_{++}^K$
		\Repeat
		\For{$k \in \set{K}$}
			\State $\rvec{v}_k \gets \arg\min_{\rvec{v}_k \in \set{V}_k} \mathsf{MSE}_k(\rvec{v}_k, \vec{p})$ 
		\EndFor
		\State $\vec{p} \gets \begin{bmatrix}
			\frac{\gamma_1 p_1}{\mathsf{SINR}_1(\rvec{v}_1, \vec{p})} & \ldots & \frac{\gamma_K p_K}{\mathsf{SINR}_K(\rvec{v}_K, \vec{p})}
		\end{bmatrix}$ 
		\Until{no significant progress is observed.}
	\end{algorithmic}
\end{minipage} \\
\hline
\end{tabular}
\end{center}

Similarly, we propose to solve Problem~\eqref{eq:maxmin} via the fixed-point iterations in Proposition~\ref{prop:optimum_nuzman}(ii), as summarized below:
\begin{center}	
\begin{tabular}{|c|}
\hline
\textbf{Algorithm 2} for solving Problem~\eqref{eq:maxmin} \\
\hline
\begin{minipage}{0.8\linewidth}
	\begin{algorithmic}[1]
		\Require $\vec{p} \in \stdset{R}_{++}^K$
		\Repeat
		\For{$k\in\set{K}$}
		\State $\rvec{v}_k \gets  \arg\min_{\rvec{v}_k\in\set{V}_k}\mathsf{MSE}_k(\rvec{v}_k,\vec{p})$ 
		\EndFor
		\State $\vec{t} \gets \begin{bmatrix}
			\frac{\gamma_1p_1}{\mathsf{SINR}_1(\rvec{v}_1,\vec{p})} & \ldots & \frac{\gamma_Kp_K}{\mathsf{SINR}_K(\rvec{v}_K,\vec{p})}
		\end{bmatrix}$
		\State $\vec{p}\gets \frac{P}{\|\vec{t}\|_{\infty}}\vec{t}$ 
		\Until{no significant progress is observed.}
	\end{algorithmic}
\end{minipage} \\
\hline
\end{tabular}
\end{center}

As already discussed, convergence to an optimal solution can be established for both algorithms under the mild technical condition $(\forall k \in \set{K})~\set{V}'_k\neq \emptyset$. Clearly, the above algorithms rely on solving Problem~\eqref{eq:MSE} at each iteration, which, as already mentioned, can be done in many practical cases following \cite{miretti2021team}, and often in closed-form as a function of appropriate statistical parameters. We remark once more that $\rvec{v}_k$ is a function of the CSI, and not a deterministic vector. In practice, as we will see in the numerical examples, storing $\rvec{v}_k$ during the algorithm execution means storing the statistical parameters of the function (e.g., the regularization factor in a regularized channel inversion block). 

An important implementation aspect is the choice of the numerical technique for approximating the MSEs. One conceptually simple option is to approximate all expectations using standard Monte-Carlo methods based on a given training set of $N_\text{sim}$ channel realizations. This is also the strategy adopted in our numerical simulations in Section~\ref{sec:sim}. Although a formal analysis is left as an interesting future research direction, we anticipate that, in our simulations, a relatively small training set ($N_\text{sim}=100$) compared to the size of the network seems sufficient for ensuring statistically stable performance. Another interesting research direction is the study of more efficient alternative techniques for approximating the MSEs.
	
\section{Simulated case study}
\label{sec:sim}
As an application of our theoretical findings, this section presents an extensive numerical study of a large-scale wideband MIMO system operating in the sub-6 GHz band. The setup, detailed in Sections~\ref{ssec:setup_sim} and~\ref{ssec:info_sim}, follows closely the canonical assumptions in the cell-free massive MIMO literature \cite{demir2021}. Specifically, we consider a representative dense urban deployment featuring multiple geographically distributed APs serving multiple single-antenna users in the practical underloaded antenna regime ($M \gg K$) with partial CSI. We analyze three distinct levels of cooperation, corresponding to small-cell implementations and user-centric cell-free implementations with either centralized or distributed processing, as defined in \cite{demir2021}. After demonstrating the convergence of the algorithms in Section~\ref{ssec:convergence}, Section~\ref{ssec:comparison_info} presents an updated comparison of the three implementations in light of our main contribution. Additionally, Section~\ref{ssec:centr}, Section~\ref{ssec:distr}, and Section~\ref{ssec:cells} present a separate comparison of the three implementations against competing short-term and long-term approaches. We report that  the main conclusions of our experiments are relatively insensitive to the choice of modeling parameters, since all the schemes are optimal and hence do not exhibit unpredictable artifacts such as convergence to local optima.

\subsection{Simulation setup}\label{ssec:setup_sim}
We consider a large-scale MIMO network composed by $K=64$ users uniformly distributed within a squared service area of size $500\times 500~\text{m}^2$, and $L=16$ regularly spaced APs with $N=8$ antennas each. By neglecting for simplicity spatial correlation, we let each sub-vector $\rvec{h}_{l,k}$ of $\rvec{h}_k^{\herm} =: [\rvec{h}_{1,k}^\herm, \ldots \rvec{h}_{L,k}^\herm]$ be independently distributed as $\rvec{h}_{l,k} \sim \CN\left(\vec{0}, \beta_{l,k}\vec{I}_N\right)$, where $\beta_{l,k}>0$ denotes the channel gain between AP $l$ and user $k$. We follow the same 3GPP-like path-loss model adopted in \cite{demir2021} for a $2$ GHz carrier frequency:
	\begin{equation*}
		\beta_{l,k} = -36.7 \log_{10}\left(D_{l,k}/1 \; \mathrm{m}\right) -30.5 + Z_{l,k} -\sigma^2 \quad \text{[dB]},
	\end{equation*}
	where $D_{l,k}$ is the distance between AP $l$ and user $k$ including a difference in height of $10$ m, and $Z_{l,k}\sim \mathcal{N}(0,\rho^2)$ [dB] are shadow fading terms with deviation $\rho = 4$. The shadow fading is correlated as $\E[Z_{l,k}Z_{j,i}]=\rho^22^{-\frac{\delta_{k,i}}{9 \text{ [m]}}}$ for all $l=j$ and zero otherwise, where $\delta_{k,i}$ is the distance between user $k$ and user $i$. The noise power is 
	$\sigma^2 = -174 + 10 \log_{10}(B) + F$ [dBm],
	where $B = 20$ MHz is the bandwidth, and $F = 7$ dB is the noise figure. Under this setup, in the following sections we focus on the power minimization problem \eqref{eq:QoS} with SINR constraints corresponding to a minimum per-user rate  of $2.5$ bit/s/Hz, and on the max-min fair problem \eqref{eq:maxmin} with unitary weights and per-user power budget $P = 20$ dBm. 

\subsection{Information constraints}\label{ssec:info_sim}
For the choice of information constraints, we specialize the model in Section~\ref{ssec:info} by considering the three representative scenarios described below. First, we assume that each user $k\in \set{K}$ is served only by its $Q$ strongest APs, that is, by the subset of APs indexed by $\set{L}_k\subseteq \set{L}$, where each set $\set{L}_k$ is formed by ordering $\set{L}$ w.r.t. decreasing $\beta_{l,k}$ and by keeping only the first $Q$ elements. Furthermore, we assume each AP $l$ to acquire (and possibly exchange) local CSI $\hat{\rvec{H}}_l\eqdef[\hat{\rvec{h}}_{l,1},\ldots,\hat{\rvec{h}}_{l,K}]$,  
	\begin{equation*}
		(\forall k \in\set{K})~\hat{\rvec{h}}_{l,k} \eqdef \begin{cases}
			\rvec{h}_{l,k} & \text{if } l \in \set{L}_k,\\
			\vec{0} & \text{otherwise}.
		\end{cases}
	\end{equation*}
	This model reflects the canonical cell-free massive MIMO implementation with pilot-based local channel estimation, and possible CSI sharing through the fronthaul \cite{demir2021}; we neglect for simplicity estimation noise and pilot contamination effects, and we focus on a simple model where small-scale fading coefficients are either perfectly known at some APs or completely unknown. In contrast, we assume that the relevant long-term statistical information is perfectly shared within the network. We then study the following information constraints, corresponding to different solutions to the subproblem~\eqref{eq:MSE} solved at each iteration of the proposed algorithms.
	
	\subsubsection{Coordinated small cells} We assume each user to be served only by its strongest AP, i.e., $Q=1$, and no CSI sharing, i.e., $(\forall l \in \set{L})~S_l = \hat{\rvec{H}}_l$. For all $\vec{p}\in \stdset{R}_{++}^K$ and $k \in\set{K}$, the $l$th subvector of the optimal beamformer $\rvec{v}_k$ corresponding to the unique AP $l\in \set{L}_k$ serving user $k$ is given by the known MMSE solution \cite{massivemimobook}
	\begin{equation*}
	\rvec{v}_{l,k} = \rvec{V}^{\mathsf{LMMSE}}_l\vec{e}_k,
	\end{equation*}
	where $(\forall l\in \set{L})~\rvec{V}^{\mathsf{LMMSE}}_l \eqdef \left(\hat{\rvec{H}}_l\vec{P}\hat{\rvec{H}}_l^{\herm}+\vec{\Sigma}_l\right)^{-1}\hat{\rvec{H}}_l\vec{P}^{\frac{1}{2}}$, and $\vec{\Sigma}_l \eqdef \left(1+\sum_{i \in\{k\in\set{K} | l \notin \set{L}_k\}}\gamma_{l,i}p_i\right)\vec{I}_N$ is the local CSI error plus noise covariance matrix. Note that, while the above beamforming function depends on local CSI only, the joint optimization of its parameters is still performed centrally based on global statistical information. Therefore, this example corresponds to a tightly coordinated small-cells system, which can be used as a benchmark for more practical cellular networks.
	
	\subsubsection{Distributed cell-free} We assume each user to be jointly served by its $Q=4$ strongest APs, and no CSI sharing, i.e., $(\forall l \in \set{L})~S_l = \hat{\rvec{H}}_l$. For all $\vec{p}\in \stdset{R}_{++}^K$ and $k \in\set{K}$, the $l$th subvector of the optimal beamformer $\rvec{v}_k$ corresponding to AP $l\in \set{L}$ is the so-called local \textit{team} MMSE solution \cite{miretti2021team,miretti2021team2}
	\begin{equation}\label{eq:LTMMSE}
		(\forall l\in \set{L})~\rvec{v}_{l,k} =\rvec{V}^{\mathsf{LMMSE}}_l\vec{c}_{l,k},
	\end{equation}
	where $\vec{c}_{l,k}\in \stdset{C}^K$ is a statistical precoding stage given by the unique solution to the linear system of equations
\begin{equation*}
\begin{cases}\vec{c}_{l,k} + \sum_{j \in \set{L}_k \backslash \{l\}}\vec{\Pi}_j \vec{c}_{j,k} = \vec{e}_k & \forall l \in \set{L}_k, \\
\vec{c}_{l,k} = \vec{0}_{K\times 1} & \text{otherwise,}
\end{cases}
\end{equation*}
where $(\forall l\in \set{L})~\vec{\Pi}_l \eqdef \E\left[\vec{P}^{\frac{1}{2}}\hat{\rmat{H}}_l^\herm\rmat{V}_l\right]$. This information constraint corresponds to an enhancement of coordinated small-cells systems that leverages coherent joint processing without changing the instantaneous CSI sharing requirements.

	\subsubsection{Centralized cell-free} We assume each user to be jointly served by its $Q=4$ strongest APs, and perfect CSI sharing, i.e., $(\forall l \in \set{L})~S_l = \hat{\rvec{H}}^\herm\eqdef [\hat{\rvec{H}}_1^\herm\ldots,\hat{\rvec{H}}_L^\herm]$. For all $\vec{p}\in \stdset{R}_{++}^K$ and $k \in\set{K}$, the optimal beamformer $\rvec{v}_k$ is given by the known centralized MMSE solution \cite{demir2021} 
	\begin{equation}\label{eq:CMMSE}
\rvec{v}_k = \left(\vec{D}_k\hat{\rvec{H}}\vec{P}\hat{\rvec{H}}^\herm\vec{D}_k+\vec{D}_k\vec{\Sigma}\vec{D}_k\right)^{-1}\vec{D}_k\hat{\rvec{H}}\vec{P}^\frac{1}{2}\vec{e}_k, 
\end{equation}
where $\vec{\Sigma} := \mathrm{diag}(\vec{\Sigma}_1,\ldots,\vec{\Sigma}_L)$ is the block-diagonal global CSI error plus noise covariance matrix, and $\vec{D}_k \eqdef \mathrm{diag}(\vec{D}_{1,k},\ldots,\vec{D}_{L,k})$ is a block-diagonal matrix satisfying
\begin{equation*}
(\forall l \in\set{L})~\vec{D}_{l,k} = \begin{cases}\vec{I}_N & \text{if } l \in \set{L}_k,\\
\vec{0}_{N\times N} & \text{otherwise}.
\end{cases}
\end{equation*}
This information constraint corresponds to an enhancement of distributed cell-free systems where the beamformer of each user $k$ is computed by a cluster processor endowed with the aggregate channel $\vec{D}_k\hat{\rvec{H}}$ of size $NL \times K$.

Although neglected for simplicity, we remark that spatial correlation, channel estimation noise, pilot contamination effects, and different clustering rules could be readily incorporated in the above models following \cite{miretti2021team, miretti2023duality}. We further point out that, as shown in \cite{miretti2021team}, for the particular channel and local CSI error model considered here, \eqref{eq:LTMMSE} reduces to the known local MMSE solution with optimal LSFD design \cite{demir2021}.
\begin{figure}
		\centering
		\subfloat[]{\includegraphics[width=0.9\columnwidth]{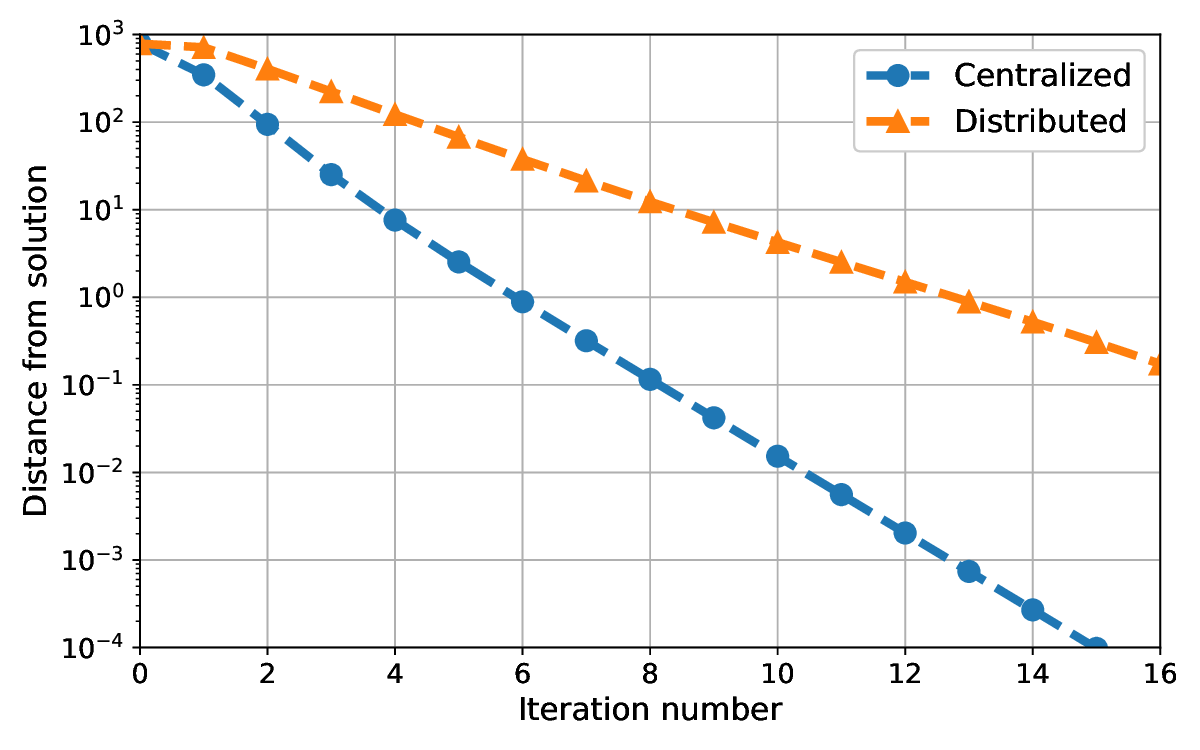}
			\label{fig:convergence_FP}}
			
		\subfloat[]{\includegraphics[width=0.9\columnwidth]{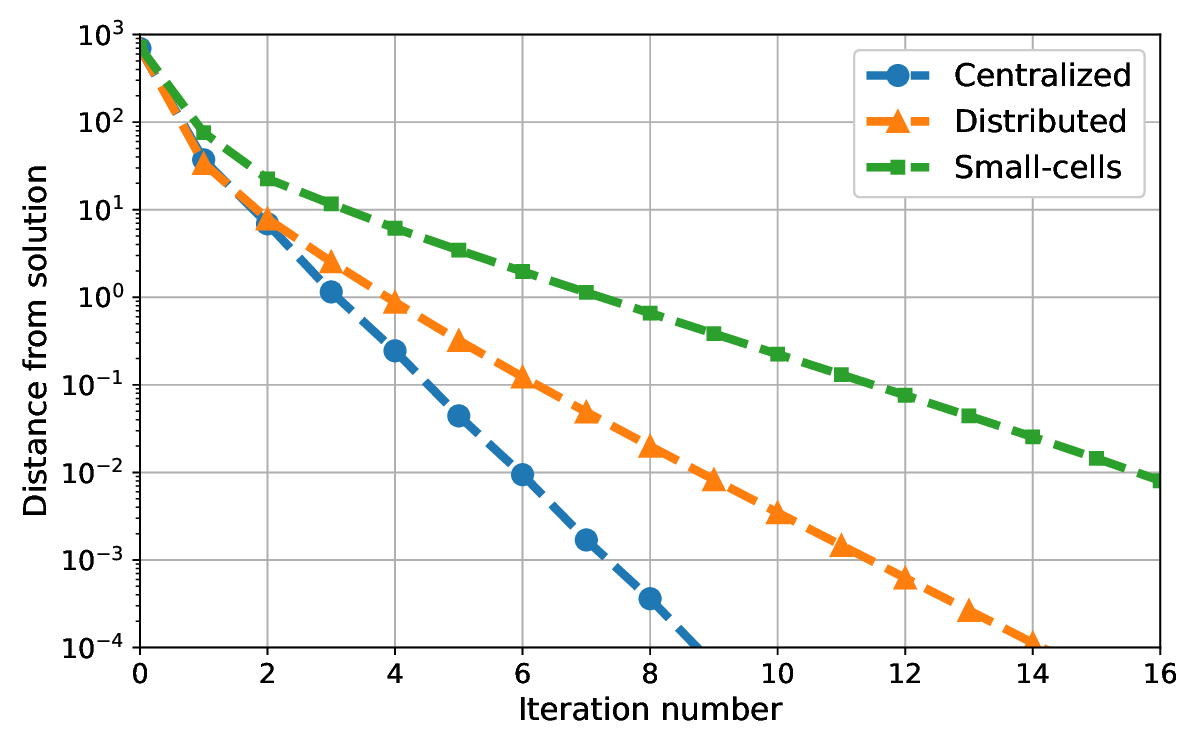}
			\label{fig:convergence_NFP}}
		\caption{Convergence behaviour of (a) the fixed-point iterations for computing a solution to Problem~\eqref{eq:QoS}; and (b) the normalized fixed-point iterations for computing a solution to Problem~\eqref{eq:maxmin} under different information constraints.}
		\label{fig:convergence}
	\end{figure}
	
	\subsection{Convergence behavior of the proposed algorithms}\label{ssec:convergence}
We run several steps of the proposed fixed-point algorithms used for computing an optimal solution to Problem~\eqref{eq:QoS} and Problem~\eqref{eq:maxmin}, for the information constraints described above and a single random user drop. In both cases, we monitor convergence to an optimal solution (if it exists) in terms of the distance $\|\vec{p}_i-\vec{p}^\star\|_2$, where $\vec{p}_i$ denotes the output of the $i$th step of the considered fixed-point algorithm, and where $\vec{p}^\star$ denotes the optimal power vector. The initialization is set to $\vec{p}_1=P\vec{1}$, and all expectations are approximated using a fixed training set of $N_{\mathsf{sim}}=100$ channel realizations. We report that, for our experiments, a higher $N_{\mathsf{sim}}$ provides indistinguishable performance.\footnote{The same observation holds for all the numerical results in this study.} It is worthwhile underlying that the implemented algorithms operate only by updating deterministic coefficients using long-term information. For example, for optimal distributed beamforming design, the implemented routines update only the long-term parameters $(\vec{c}_{l,1},\ldots,\vec{c}_{l,K},\vec{\Sigma}_l)_{l\in \set{L}}$.
	
Figure~\ref{fig:convergence_FP} reports the convergence behavior of the fixed-point iterations for computing an optimal solution to Problem~\eqref{eq:QoS}. Due to unfeasible minimum rate constraints, the curve related to the small-cells network  shows a divergent behaviour and hence it is omitted. In contrast, since the minimum rate constraints are feasible for both the centralized and distributed cell-free networks, the corresponding fixed-point iterations converge to an optimal solution, as expected. Furthermore, we observe that convergence is geometric in both cases, as predicted by the discussion in Appendix~\ref{app:yates}. However, we notice a faster convergence for the centralized case. This can be explained using the results in \cite{cavalcante2018spectral,Piotrowski2022}, which show that bounds on the convergence speed can be directly connected to a well-defined notion of spectral radius for non-linear mappings. Informally, in our setup, this notion is related to achievable rates in the interference limited regime. Therefore, the different convergence speed can be explained by the fact that, due to better interference management capabilities, the rates in the centralized case saturate at a larger value than in the distributed case as the SNR increases. The interested reader is referred to \cite{cavalcante2018spectral,Piotrowski2022} for additional details.
	
	Figure~\ref{fig:convergence_NFP} reports the convergence behavior of the normalized fixed-point iterations for computing an optimal solution to Problem~\eqref{eq:maxmin}. In contrast to the unnormalized fixed-point iterations, the normalized fixed-point iterations always converge. Similarly to Figure~\ref{fig:convergence_FP}, we observe geometric convergence as predicted by the discussion in  Appendix~\ref{app:nuzman}. Furthermore, we observe a similar difference in convergence speed across the different information constraints. This difference seems to depend again on the different interference management capabilities as in Figure~\ref{fig:convergence_FP}. A formal estimation of the convergence speed could be obtained, for example, using the bounds in \cite[Remark~2.1.12]{krause2015}. However, we point out that, in contrast to the unnormalized fixed-point iterations, a rigorous connection between convergence speed and interference management capabilities is not reported in the literature and requires further analysis.  

\begin{figure}
		\centering
		\subfloat[]{\includegraphics[width=0.9\columnwidth]{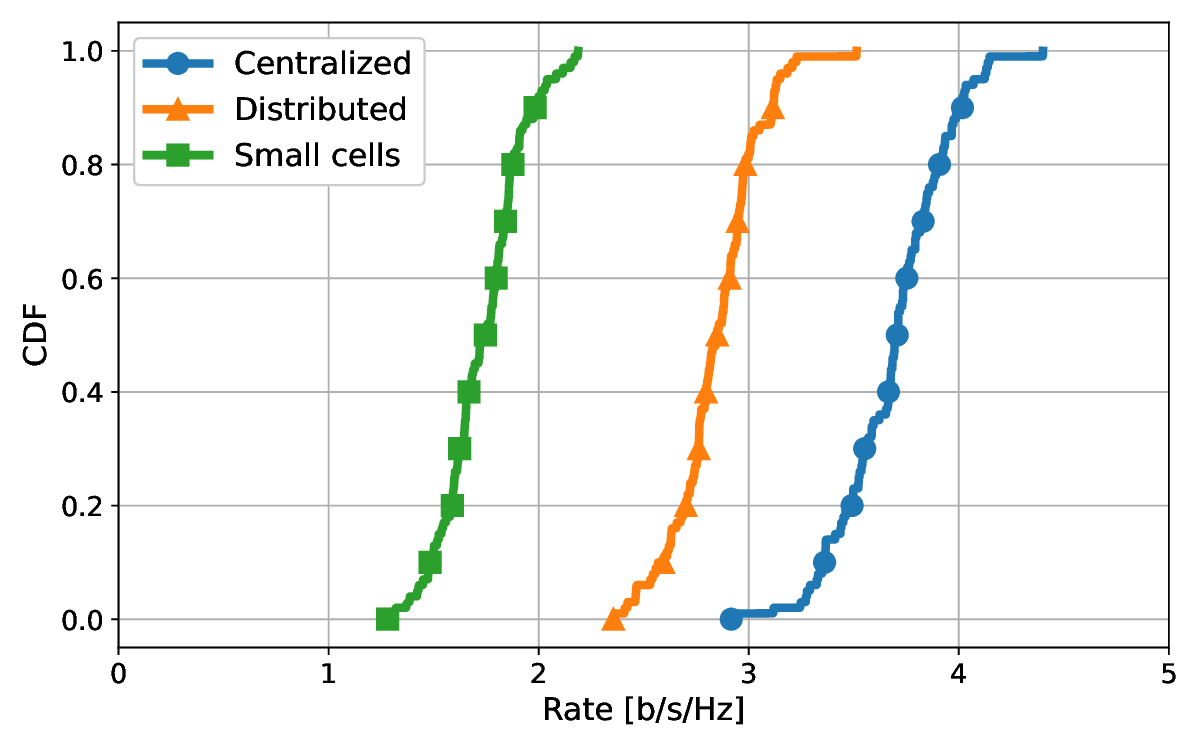}
		\label{fig:CDFUatF}}
	
		\subfloat[]{\includegraphics[width=0.9\columnwidth]{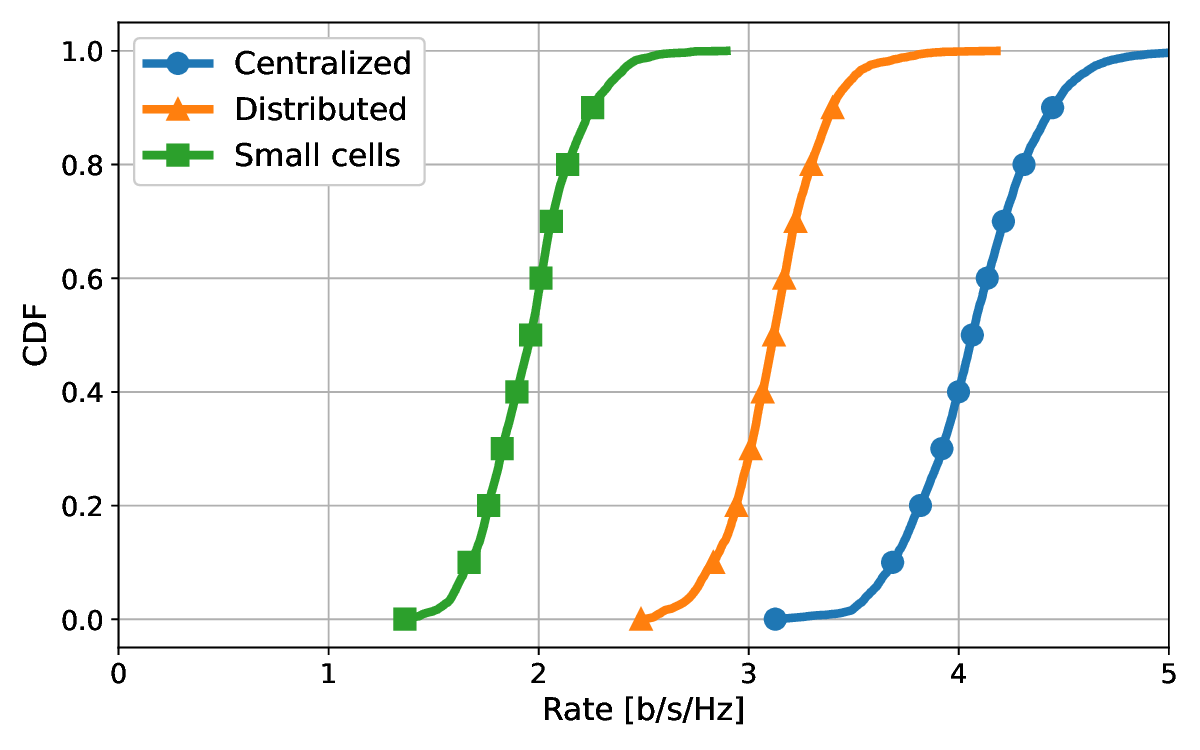}
		\label{fig:CDFcoh}}
		\caption{Empirical CDF of the ergodic rates achieved by the proposed optimal solution to Problem~\eqref{eq:maxmin} under multiple user drops and different information constraints, evaluated using (a) the UatF bound in \eqref{eq:uatf}; and (b) its upper bound \eqref{eq:coh} based on coherent decoding.}
		\label{fig:CDF}
	\end{figure}
\subsection{Comparison of cell-free and small-cells networks}\label{ssec:comparison_info}
We now show how our theoretical results can be applied to revisit the available studies comparing cell-free networks and small-cells networks in delivering uniformly good quality of service to all users (i.e, in the max-min sense), such as \cite{ngo2017cell,bjornson2019making,demir2021}. One main limitation of the available studies is that, being based on either suboptimal power control or suboptimal beamforming design, they do not truly characterize the achievable max-min fair performance of the considered networks and hence they may lead to misleading conclusions.

In particular, we solve the considered max-min fair problem~\eqref{eq:maxmin} with unitary weights for $N_{\mathsf{CDF}}=100$ independent random user drops. Then, we report in Figure~\ref{fig:CDF} the empirical cumulative density function (CDF) of the resulting per-user rates given by the UatF bound \eqref{eq:uatf}. As we can see, for the considered setup, cell-free networks significantly outperform small-cell networks for all percentiles. In addition, although to a lesser extent, the centralized information structure outperforms the distributed one for all percentiles. In contrast, most of the available studies focusing on max-min fairness show similar CDF curves crossing at some given percentile. However, our results show that this is an artifact of the chosen power control and beamforming design, which are not jointly optimized for max-min fairness. Note that a similar observation holds also by setting arbitrary weights in Problem~\eqref{eq:maxmin}. In fact, by the principle that adding information (in our context, replacing an information constraint $\set{V}_k$ with a superset $\set{V}'_k \supseteq \set{V}_k$) cannot hurt, the achievable rate region of cell-free networks is a superset of the cellular one.

As already discussed in Remark~\ref{rem:bounds}, one potential limitation of the above analysis is that it is based on the UatF bound which, despite offering tractable problem formulations and hence leading to the proposed algorithms, may underestimate the true system performance. However, we now demonstrate experimentally that the UatF bound can be a quite effective proxy for optimizing the system according to less pessimistic but intractable performance metrics. More specifically, we report in Figure~\ref{fig:CDF} the performance of the proposed solution to Problem~\eqref{eq:maxmin} in terms of the following classical ergodic rate expressions based on coherent decoding: $(\forall k \in \set{K})$
\begin{equation}\label{eq:coh}
	R_k^{\mathsf{coh}}(\rvec{v}_k,\vec{p}) \eqdef \E\left[\log_2\left(1+\dfrac{p_k|\rvec{h}_k^\herm\rvec{v}_k|^2}{\textstyle\sum_{j\neq k} p_j|\rvec{h}_j^\herm\rvec{v}_k|^2+\|\rvec{v}_k\|_2^2}\right)\right].
\end{equation}
Interestingly, we observe that the general trend of the curves in Figure~\ref{fig:CDFUatF} is maintained in Figure~\ref{fig:CDFcoh}. In paricular, we notice that the gap between the CDFs in Figure~\ref{fig:CDFUatF} and the corresponding ones in Figure~\ref{fig:CDFcoh} is below $0.4$ b/s/Hz for nearly all percentiles. 

\subsection{Comparison of centralized cell-free schemes}
\label{ssec:centr}
\begin{figure}
		\centering
		\subfloat[]{\includegraphics[width=0.9\columnwidth]{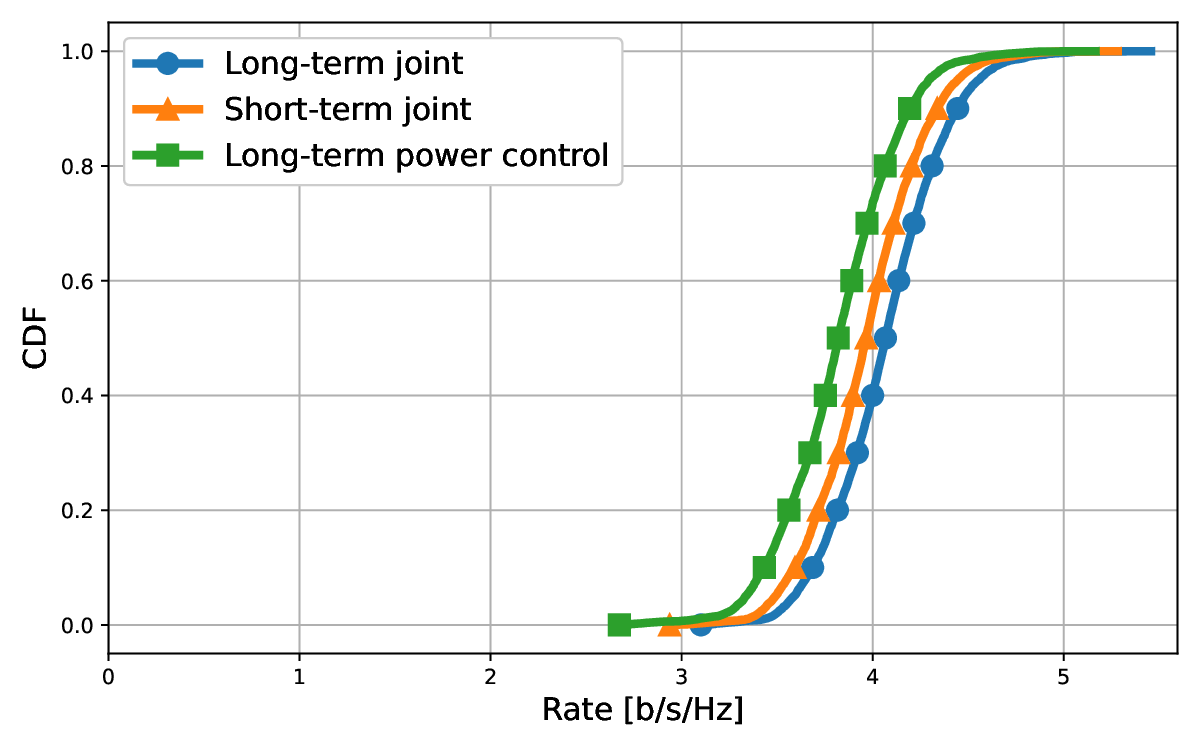}
		\label{fig:comparison_centr}}
		
		\subfloat[]{\includegraphics[width=0.9\columnwidth]{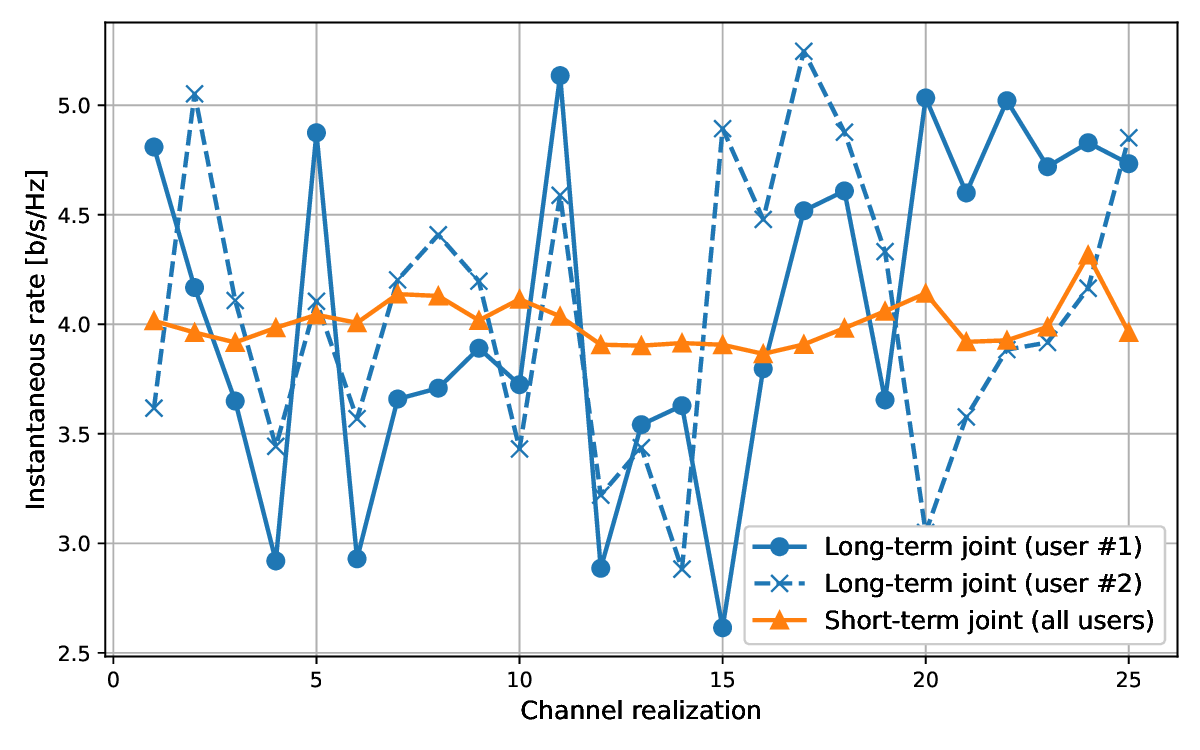}
		\label{fig:comparison_fluct}}
		\caption{Comparison of centralized cell-free schemes: (a) Empirical CDF of the ergodic rates \eqref{eq:coh} achieved by the proposed optimal solution to Problem~\eqref{eq:maxmin} under a centralized cell-free information constraint, and by the competing techniques in Section~\ref{ssec:centr}; and (b) instantaneous rates (as defined for the short term approach) across users and channel realizations.}
\label{fig:comparison_sota}
	\end{figure}
	In this section we compare different methods for max-min fair system design in centralized cell-free networks. In particular, we compare the performance of the proposed algorithm for computing an optimal solution to Problem~\eqref{eq:maxmin} under a centralized information constraint against the following classical techniques:
	\subsubsection{Short-term joint power control and beamforming design}
	 Following a similar approach as the one in \cite{schubert2004solution} for cellular networks with perfect CSI, for each realization $\hat{\vec{H}} \reqdef [\hat{\vec{h}}_1,\ldots,\hat{\vec{h}}_K]$ of the global CSI $\hat{\rvec{H}}$, we solve
	\begin{equation*}
	\begin{aligned}	\underset{\substack{\vec{p}\in \stdset{R}_{+}^K \\ \vec{v}_1,\ldots,\vec{v}_K \in \stdset{C}^{M} }}{\text{maximize}}
		& \quad  \min_{k\in\set{K}} \dfrac{p_k|\hat{\vec{h}}_k^\herm\vec{D}_k\vec{v}_k|^2}{\textstyle\sum_{j\neq k} p_j|\hat{\vec{h}}_j^\herm\vec{D}_k\vec{v}_k|^2+\vec{v}_k^\herm\vec{D}_k\vec{\Sigma}\vec{D}_k\vec{v}_k}\\
		\text{subject to} & \quad \|\vec{p}\|_{\infty}\leq P \\
		& \quad  (\forall k \in \set{K})~\|\vec{v}_k\|_2^2 \neq 0,
	\end{aligned}
\end{equation*}
where the parameters $\vec{\Sigma}$ and $\vec{D}_k$ are the same as in \eqref{eq:CMMSE}. The above objective corresponds to the instantaneous SINR for a variant of the ergodic rate expression in \eqref{eq:coh} that takes into account channel estimation errors \cite[Theorem~5.1]{demir2021}. The optimal beamformers take the same form of \eqref{eq:CMMSE} (maximization of a Rayleigh quotient), with the difference that $\vec{p}$ is a function of the instantaneous channel realization $\hat{\vec{H}}$. The solution can be computed by adapting the block-coordinate ascent method in \cite{schubert2004solution}, or via the normalized-fixed point iterations in  Appendix~\ref{app:nuzman}.
\subsubsection{Long-term power control with fixed beamforming design}
In this case, we solve 
\begin{equation*}
	\begin{aligned}
		\underset{\substack{\vec{p}\in \stdset{R}_{+}^K}}{\text{maximize}}
		& \quad  \min_{k\in\set{K}} \mathsf{SINR}_k(\rvec{v}^{\mathsf{fix}}_k,\vec{p}) \\
		\text{subject to} & \quad \|\vec{p}\|_{\infty}\leq P,
	\end{aligned}
\end{equation*}
where $\rvec{v}^{\mathsf{fix}}_k$ denotes a fixed beamforming design for each user~$k\in \set{K}$. As done in \cite{demir2021}, we choose the centralized MMSE solution \eqref{eq:CMMSE} by  assuming full power transmission, i.e., we choose the solution to \eqref{eq:maxSINR} for $\vec{p}=\vec{1}P$. The above problem can be solved via the normalized-fixed point iterations in Appendix~\ref{app:nuzman}, as in \cite{demir2021}, or in closed-form using \cite{miretti2022vtc}. 

Figure~\ref{fig:comparison_centr} reports the empirical CDF over $N_{\mathsf{CDF}}=100$ independent user drops of the resulting per-user rates given by \eqref{eq:coh}. As we can see, at the cost of increased complexity in solving the underlying optimization problem (but not in computing the beamformers), long-term joint power control and beamforming design may provide non-negligible performance gains with respect to long-term power control only. Perhaps more surprisingly, we also observe that long-term optimization may perform better than short-term optimization. In addition to performance gains, we recall that the long-term approach may also offer  significant gains in terms of computational complexity, since the optimization problem need not be solved for each channel realization as in the short-term approach.

The performance degradation of the short-term approach compared to the long-term approach may seem counterintuitive, as the former can optimally adapt uplink powers to the instantaneous CSI, whereas the latter cannot. However, this advantage can be (and often is) outweighed by the following disadvantage related to the max-min fairness criterion. As illustrated in Figure~\ref{fig:comparison_fluct}, the short-term approach assigns the same instantaneous rate to all users in each channel realization. In contrast, the long-term approach assigns only the same \textit{ergodic} rate to all users, allowing for greater flexibility in adapting instantaneous rates across users and realizations. For instance, long-term optimization is less sensitive to channel realizations with a single deep-fade event for a specific user, which would otherwise severely limit the instantaneous performance of all users in the short-term approach.

\subsection{Comparison of distributed cell-free schemes}
\label{ssec:distr}
\begin{figure}
		\centering
		\includegraphics[width=0.9\columnwidth]{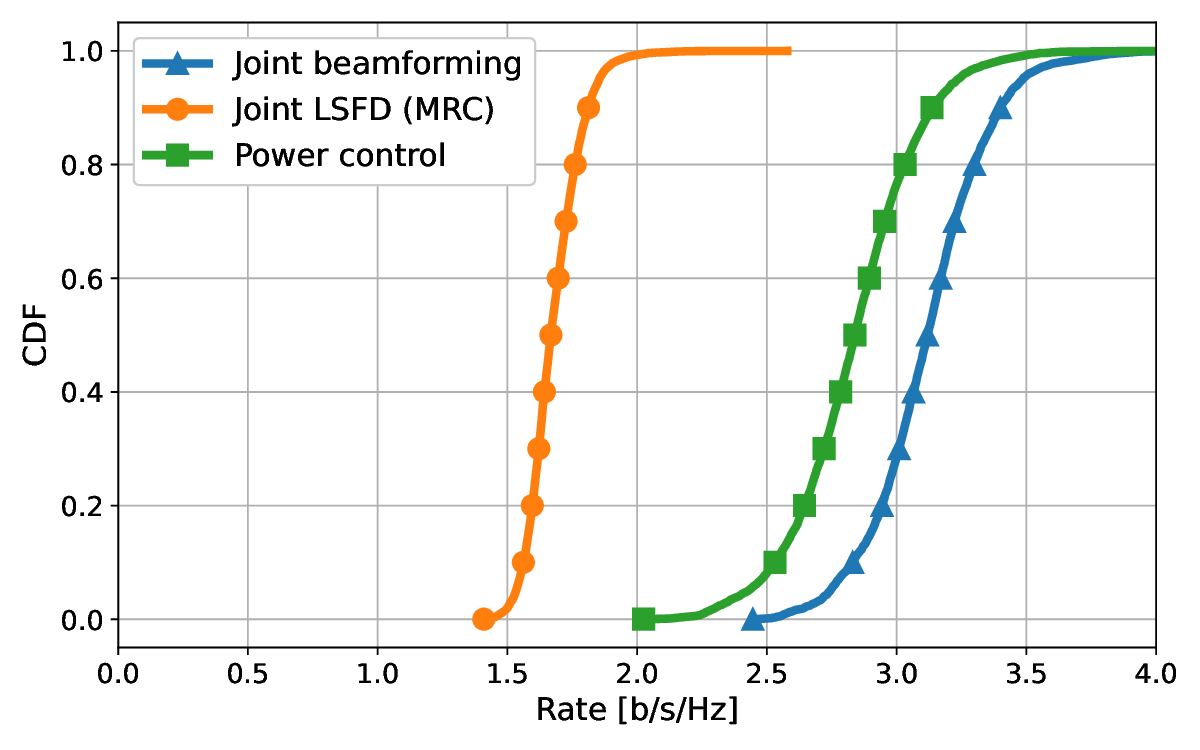}
		\caption{Comparison of distributed cell-free schemes. Empirical CDF of the ergodic rates \eqref{eq:coh} achieved by the proposed optimal solution to Problem~\eqref{eq:maxmin} under a distributed cell-free information constraint, and by the competing techniques in Section~\ref{ssec:distr}.}
		\label{fig:comparison_distr}
\end{figure}
We now compare different methods for long-term max-min fair system design in distributed cell-free networks. The short-term approach is not considered since, to the best of our knowledge, it cannot be easily generalized to optimal beamforming design under distributed information constraints. In particular, we compare the performance of the proposed algorithm for computing an optimal solution to Problem~\eqref{eq:maxmin} under a distributed information constraint against the following state-of-the-art techniques:
 
\subsubsection{Long-term power control with fixed beamforming design}
Similarly to the centralized information case, we use the normalized fixed-point iterations in Appendix~\ref{app:nuzman} as in \cite{demir2021}, or the closed-form expression in \cite{miretti2022vtc}, to solve
\begin{equation*}
	\begin{aligned}
		\underset{\substack{\vec{p}\in \stdset{R}_{+}^K}}{\text{maximize}}
		& \quad  \min_{k\in\set{K}} \mathsf{SINR}_k(\rvec{v}^{\mathsf{fix}}_k,\vec{p}) \\
		\text{subject to} & \quad \|\vec{p}\|_{\infty}\leq P,
	\end{aligned}
\end{equation*}
where $\rvec{v}^{\mathsf{fix}}_k$ denotes a fixed beamforming design for each user~$k\in \set{K}$ satisfying the distributed information constraints. We choose the optimal solution to \eqref{eq:maxSINR} for  $\vec{p}=\vec{1}P$, i.e., we choose the local team MMSE (LTMMSE) solution \eqref{eq:LTMMSE} assuming full power transmission.  

\subsubsection{Long-term joint power control and LSFD design}
Following the approach in \cite{bashar2019maxmin}, we solve 
\begin{equation*}
	\begin{aligned}
		\underset{\substack{\vec{p}\in \stdset{R}_{+}^K \\ \vec{a}_1,\ldots,\vec{a}_K \in \stdset{C}^{L} }}{\text{maximize}}
		& \quad  \min_{k\in\set{K}} \mathsf{SINR}_k(\mathrm{diag}(a_{1,k}\vec{I}_N,\ldots,a_{L,k}\vec{I}_N)\rvec{v}^{\mathsf{fix}}_k,\vec{p}) \\
		\text{subject to} & \quad \|\vec{p}\|_{\infty}\leq P,
	\end{aligned}
\end{equation*}
where $\rvec{v}^{\mathsf{fix}}_k$ denotes a fixed beamforming design for each user~$k\in \set{K}$, and $\vec{a}_k$ denotes the corresponding LSFD coefficients as defined in \cite{demir2021}. As in \cite{bashar2019maxmin}, we choose maximum ratio combining  (MRC) $\rvec{v}^{\mathsf{fix}}_k = \hat{\rvec{h}}_k$. Under this choice of beamformers, \cite{bashar2019maxmin} proves that the above problem can be solved via a block coordinate ascent algorithm, where the LSFD coefficients are updated in closed-form by maximizing a Rayleigh quotient, and where the power control vector is updated by solving a long-term power control problem with fixed beamforming design. However, we remark that the normalized fixed-point iterations in Appendix~\ref{app:nuzman} can also be used to solve the above problem as well as more general versions not limited to MRC. 

Figure~\ref{fig:comparison_distr} reports the performance of the considered methods for distributed cell-free networks. Similarly to Figure~\ref{fig:comparison_centr}, we observe that joint power control and beamforming design may provide non-negligible gains with respect to power control only, and  that the gap appears larger than in the centralized case. Furthermore, we notice that beamforming design is clearly the main limiting factor, since considering optimal LSFD design seems not sufficient to recover the large performance loss of MRC with respect to the local team MMSE scheme in \eqref{eq:LTMMSE}.

\subsection{Comparison of small-cells schemes}\label{ssec:cells}
\begin{figure}
		\centering
		\includegraphics[width=0.9\columnwidth]{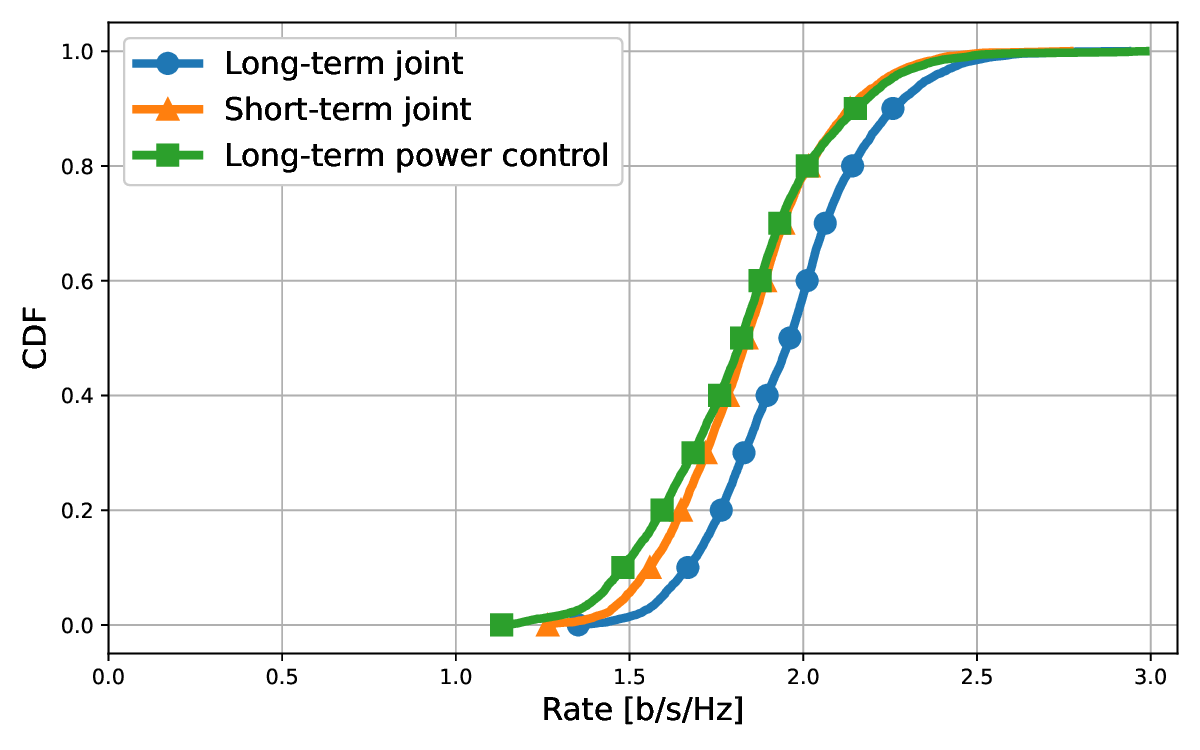}
		\caption{Comparison of small-cell schemes. Empirical CDF of the ergodic rates \eqref{eq:coh} achieved by the proposed optimal solution to Problem~\eqref{eq:maxmin} under a small-cells information constraint, and by the competing techniques in Section~\ref{ssec:centr} adapted to the case of $Q=1$.}
		\label{fig:comparison_smallcells}
\end{figure}
To conclude the case study, we compare different methods for max-min fair system design in the case of coordinated small-cells systems. In particular, we repeat the experiments in Section~\ref{ssec:centr} by adapting all the considered optimization problems and corresponding schemes to the case of $Q=1$ (each user is only served by its strongest AP). We recall that, for small-cells systems, no distinction is made between centralized and distributed beamforming since these notions relate to coherent joint processing of multiple APs.
 
The results are reported in Figure~\ref{fig:comparison_smallcells}. As for the centralized cell-free case, we observe that long-term optimization provides better performance while significantly reducing complexity and overhead. The benefits of the long-term schemes in terms of complexity and overhead is particularly evident for coordinated small-cells systems, since all beamformers can be computed locally based on local CSI, whereas all the coordination is based on slowly-varying global statistical information. In contrast, the short-term scheme still requires real-time sharing and processing of global instantaneous CSI. Finally, as expected, we observe that, at the cost of increased complexity, long-term joint power control and beamforming design performs better than long-term power control.

\section{Conclusion and future directions}
This study demonstrates that the axiomatic power control framework in \cite{yates95} \cite{nuzman07} and the beamforming design method in \cite{miretti2021team} can be combined to solve  challenging (infinite dimensional) two-timescale versions of the classical joint power control and beamforming problems \eqref{eq:QoS_det} and \eqref{eq:maxmin_det}, yielding novel optimal long-term optimization methods with favorable properties for modern large-scale MIMO systems. Due to the generality of our approach, the proposed algorithms could be extended in future works to more involved problems, e.g., by incorporating optimal user association as in \cite{schubert2019multi,liu2019association} or innovative information constraints as in \cite{miretti2021team,miretti2021team2,miretti2024delayed}. Finally, a point left open by this study is the evaluation of the convergence properties of alternative methods for the max-min fair criterion based on block-coordinate ascent similar to \cite{schubert2004solution} and \cite{bashar2019maxmin}. In fact, Proposition~\ref{prop:MSE} readily enables the derivation of a novel block-coordinate ascent algorithm for \eqref{eq:maxmin}. However, despite some promising experimental evidence \cite{miretti2022joint}, establishing its optimality may be nontrivial, since the arguments  in \cite{schubert2004solution} and \cite{bashar2019maxmin} seem not directly applicable.
		
\appendix[Axiomatic framework for power control]
\label{sec:power}
In this appendix we provide an up-to-date introduction to the axiomatic framework for power control in \cite{yates95,nuzman07}. In addition to offering a concise self-contained background for the main results of this study, this appendix may be of independent interest, since it uncovers overlooked connections within the existing literature. In more detail, we investigate networks consisting of $K$ users denoted by the set $\mathcal{K}:=\{1,\ldots,K\}$. To each user $k\in\mathcal{K}$ we associate an utility function $u_k:\real^K_{+}\to\real_{+}$, and we assume that all functions $(u_k)_{k\in\mathcal{K}}$ take a common argument $\vec{p}=(p_1,\ldots,p_K)\in\real_+^K$, which we interpret as the vector of powers of all users. We consider power control problems involving the above utilities and a given monotone norm $\|\cdot\|$, e.g., to impose power constraint of the type $(\forall \vec{p}\in\real_+^K)~\|\vec{p}\|\le P$, where $P$ is a given constant. To keep the utility optimization problems tractable while maintaining the framework reasonably general, we need to restrict the class of utility functions. To this end, we leverage the concept of standard interference functions, defined as follows: 

\begin{definition}
	\label{definition.mappings} \cite{yates95} A function $f:\real^K_+\to\real_{++}$ is said to be a \emph{standard interference function} if the following properties hold:  \begin{enumerate}[(i)]
		\item $(\forall \vec{p}\in\real^K_{+})(\forall \vec{q}\in\real^K_{+}) ~ \vec{p}\ge\vec{q} \Rightarrow f(\vec{p})\ge f(\vec{q})$; and \par 
		\item $(\forall \vec{p}\in\real^K_+)$ $(\forall \alpha>1)$  $\alpha {f}(\vec{p})>f(\alpha\vec{p})$. \par 
	\end{enumerate}
	Likewise, a mapping $T:\real_+^K\to\real_{++}^K:\vec{p}\mapsto (f_1(\vec{p}),\ldots,f_N(\vec{p}))$ is said to be a standard interference (SI) mapping if each coordinate function $f_k:\real^K_+\to\real_{++}$ ($k=1,\ldots,K$) is a standard interference function.
\end{definition}

\begin{remark}
	Standard interference functions $f:\real^K_+\to\real_{++}$ are continuous on $\real_{++}^K$, but they are not necessarily continuous on the boundary $\real^K_+\backslash \real_{++}^K$. Nevertheless, for all practical purposes, hereafter we can tacitly assume continuity. The reason is that the restriction  $f|_{\real_{++}^K}:\real^K_{++}\to\real_{++}:\vec{p}\mapsto f(\vec{p})$ of $f$ to ${\real_{++}^K}$ can be continuously extended to $\real_{+}^K$ \cite[Theorem~5.1.5]{Lemmens2012}. 
\end{remark}
Identifying standard interference functions directly from the definition is not necessarily easy. However, the next proposition provides a simple technique that can assist us in this task.
\begin{proposition} \label{proposition.concavity} Let $f:\real^K_+\to\real_{++}$ be a (positive) concave function (not necessarily everywhere continuous). Then $f$ is a standard interference function (see \cite[Proposition~1]{renato2016} for a simple proof of this standard result). 
\end{proposition} 

For later reference, we say that a mapping $T:\real_+^K\to\real_{++}^K$ is a positive concave mapping if it can be written as $(\forall\vec{p}\in\real_+^K)~ T(\vec{p})=(f_1(\vec{p}), \ldots, f_K(\vec{p}))$, where, for every $k\in\set{K}$, $f_k:\real^K_+\to\real_{++}$ is a (positive) concave function. We now have in place all the required definitions to describe the utility functions covered by the framework.

\begin{assumption}
	\label{assumption.restrictions}
	The utility function $u_k$ of each user $k\in\mathcal{K}$ takes the form  
	$
		u_k:\real_+^K\to\real_+:\vec{p}\mapsto\dfrac{p_k}{f_k(\vec{p})},
	$
	where $f_k:\real_+^K\to\real_{++}$ is a \emph{continuous} standard interference function. 
\end{assumption}

\subsection{Minimization of the sum transmit power}
\label{app:yates}
In this section, we describe the utility optimization problems first considered in the wireless domain in the seminal study in \cite{yates95}. We also summarize  recent results and connections to fixed point theory in Thompson metric spaces. These connections seem to be largely ignored in the wireless literature, but, as shown in \cite{nuzman07,Piotrowski2022}, their knowledge enables us to (i) gain powerful mathematical machinery to answer many questions related to power control problems in a unified way and to (ii) give alternative proofs of some results in \cite{yates95}. In more detail, the first power control problem we describe has the objective of minimizing the sum transmit power subject to minimum utility requirements. More precisely, we consider problems of the form in \eqref{eq:yates}, where the utility functions $(u_k)_{k\in\mathcal{K}}$ satisfy Assumption~\ref{assumption.restrictions}. We now proceed to review answers to three important questions related to the optimization problem in \refeq{eq:yates}:
(Q1) Does a solution exist? (Q2) Is the solution unique? (Q3) Are there simple algorithms able to compute the solution?

To answer the above questions, we first reformulate \refeq{eq:yates} as a fixed point problem. To this end, we can use well-known properties of the utility functions satisfying Assumption~\ref{assumption.restrictions}  (see, for example, \cite[Lemma~1]{cavalcante2023})  and standard arguments in the literature to verify that, if a solution to \refeq{eq:yates} exists, all constraints have to be satisfied with equality. Therefore, assuming $\vec{p}^\star=(p_1^\star,\ldots,p_K^\star)$ to be a solution to \refeq{eq:yates}, we deduce from Assumption~\ref{assumption.restrictions}:
$
(\forall k\in\mathcal{K})~ u_k(\vec{p}^\star) = \gamma_k \Leftrightarrow (\forall k\in\mathcal{K})~p_k^\star = \gamma_k f_k(\vec{p}^\star) \Leftrightarrow \vec{p}^\star \in\mathrm{Fix}(T)$,
where $T:\real_{+}^K\to\real_{++}^K:\vec{x}\mapsto (\gamma_1 f_1(\vec{x}), \ldots, \gamma_K f_K(\vec{x}))$ is a standard interference mapping. 
These relations reveal that properties of  solutions to \refeq{eq:yates} are connected to properties of the fixed point set of the standard interference mapping $T$. 

Recalling that standard interference mappings are only guaranteed to be contractive in the cone $\real_{++}^K$ with respect to Thompson's metric \cite[Lemma~2.1.7]{Lemmens2012}\cite{nuzman07} (but not necessarily Lipschitz contractions),  we conclude that $\mathrm{Fix}(T)$ is either a singleton or the empty set \cite[Ch.~4.1]{krause2015}, so we now have a complete answer to questions (Q1) and (Q2): the solution to the problem in \refeq{eq:yates} is unique \emph{if it exists}. Furthermore, this answer enables us to reformulate \refeq{eq:yates} equivalently as follows:
\begin{align}
	\label{eq.fixed_point_yates}
	\text{Find }\vec{p}^\star\in\real_{++}^K\text{ such that } \vec{p}^\star\in \mathrm{Fix}(T),
\end{align}
and we refer readers to \cite[Proposition~4]{renato2019} for a necessary and sufficient condition for the existence of the fixed point of $T$, and, hence, for the existence of the necessarily unique solution to \refeq{eq:yates}. This result in \cite[Proposition~4]{renato2019} can frequently simplify the analysis of existence of solutions to \refeq{eq:yates} \cite{renato2019}.

An additional advantage of the reformulation in \refeq{eq.fixed_point_yates}, which is a standard fixed point problem involving contractive mappings in Thompson metric space, is that we may conjecture that the sequence $(\vec{p}_n)_{n\in\Natural}$  generated via 
\begin{align}
	\label{eq.banach_picard}
	\vec{p}_{n+1}=T(\vec{p}_n),
\end{align}
converges for any starting point $\vec{p}_1\in\real_+^K$ (assuming that $\mathrm{Fix}(T)\neq\emptyset$), and this conjecture can be settled by using basic arguments in analysis \cite{yates95} or by using  arguments based on fixed point theory in metric spaces \cite{nuzman07}. For reference, we summarize the above results below.

\begin{proposition}\label{prop:convergence_yates} \cite{yates95,nuzman07}
	Let $T:\real_+^K\to\real_{++}^K$ be a standard interference mapping. Then $\mathrm{Fix}(T)$ is a singleton if not empty. Furthermore, if $\emptyset\neq\mathrm{Fix}(T)\ni\{\vec{p}^\star\}$, then the sequence $(\vec{p}_n)_{n\in\Natural}$ generated via \refeq{eq.banach_picard} converges in norm to $\vec{p}^\star$ for any starting point $\vec{p}_1\in\real_+^K$.
\end{proposition}

One potential limitation of the fixed point iteration in \refeq{eq.banach_picard} is that it can be slow. Indeed, \cite[Example~2]{fey2012} shows that convergence can be sublinear in general. Nevertheless, if the standard interference functions $(f_k)_{k\in\mathcal{K}}$ are further restricted to be positive and concave, as in this study, then the convergence is guaranteed to be geometric \cite{Piotrowski2022} (in the sense of \cite[Definition~1]{Piotrowski2022}) and no better than geometric in general \cite[Proposition~5]{cavalcante2018spectral}. In this case, we can further improve the convergence speed with, for example, the acceleration technique described in \cite{renato2016} or with the more computationally intensive but faster approach in \cite[Theorem~11]{boche2008}. (For geometric convergence of fixed point iterations of standard interference mappings that are not necessarily concave, we refer the readers to \cite{cavalcante2018spectral}.)

\subsection{Weighted max-min fair power control}
\label{app:nuzman}
In this section, we consider optimization problems that have the objective of maximizing the minimum weighted utility across all users subject to power constraints. More formally, we consider power control problems of the form in \eqref{eq:nuzman}, by assuming again that the utility functions $(u_k)_{k\in\mathcal{K}}$ satisfy Assumption~\ref{assumption.restrictions}. 

As shown in the next proposition, the optimization problem in \refeq{eq:nuzman} always has a solution that can be obtained with fixed point iterations involving mappings that are contractive in a (compact) Hilbert-projective metric space. To the best of our knowledge, in the wireless domain, this connection to fixed point theory has been first established in \cite{nuzman07}, and, with slightly different assumptions on the utility functions and/or the power constraints, these fixed point algorithms have also been independently proposed in \cite{cai2011maxmin,tan2014wireless,zheng2016} \cite[Ch.~4]{tan2015wireless}, which are studies building upon the results in \cite{krause1986perron}. We emphasize that, unlike the properties imposed on the utility functions in some of these previous studies, the properties in Definition~\ref{definition.mappings} are not enough to guarantee uniqueness of the solution to \refeq{eq:nuzman}, so the characterization in Proposition~\ref{prop:fixed_point}(iii) is important. We omit the proof of Proposition~\ref{prop:fixed_point} because it is immediate from \cite[Proposition~2]{cavalcante2023}, which builds upon \cite{nuzman07}, and it can be also obtained with the results in \cite{krause1986perron}.

\begin{proposition}\label{prop:fixed_point}
	Consider the problem in \refeq{eq:nuzman}, and recall that $(\forall k\in\mathcal{K})(\forall \vec{p}\in\real^K_+)~u_k(\vec{p})=p_k/f_k(\vec{p})$ (Assumption~\ref{assumption.restrictions}). Let $T:\stdset{R}_+^K\to \stdset{R}_{++}^K$ be the standard interference mappings defined by $(\forall\vec{p}\in\real_{++}^K)~T(\vec{p}):=(\gamma_1f_1(\vec{p}), \ldots,\gamma_Kf_K(\vec{p}))$. Then each of the following holds:
	\begin{enumerate}[(i)]
		\item The set $\mathrm{Fix}(\tilde{T})$ of fixed points  of the mapping 
		$
			\tilde{T}:  \stdset{R}_{+}^K \to \stdset{R}_{++}^K  : \vec{p} \mapsto \frac{P}{\|T(\vec{p})\|}T(\vec{p})
		$
		is a singleton, and  its unique fixed point $\vec{p}^\star\in\mathrm{Fix}(\tilde{T})$  is a (not necessarily unique) solution to \eqref{eq:nuzman}.
		\item For every $\vec{p}_1 \in \stdset{R}_{+}^K$, the sequence $(\vec{p}_n)_{n\in\stdset{N}}$ generated via 
		\begin{equation}
			\label{eq:italg}
			(\forall n \in \stdset{N})~\vec{p}_{n+1} = \tilde{T}(\vec{p}_n)
		\end{equation}
		converges in norm to the fixed point $\vec{p}^\star\in\mathrm{Fix}(\tilde{T})$.
		\item Denote by $\set{P}\subset\stdset{R}_{++}^K$ the nonempty set of solutions to the problem in~\eqref{eq:nuzman}. Then the fixed point $\vec{p}^\star\in\mathrm{Fix}(\tilde{T})$ satisfies $(\forall \vec{p}\in \set{P})~\vec{p}\ge \vec{p}^\star$.
	\end{enumerate}
\end{proposition}

Similarly to the fixed point iteration in \refeq{eq.banach_picard}, the fixed point algorithm in \refeq{eq:italg} also converges geometrically (in the sense of \cite[Definition~1]{Piotrowski2022}) if the mapping $T:\real^K_+\to\real_{++}^K$ is a (positive) concave mapping \cite[Remark~2.1.12]{krause2015}.
	
\bibliographystyle{IEEEbib}
\bibliography{IEEEabrv,refs}

\end{document}